\documentclass[final,11pt]{amsart}
\RequirePackage[utf8]{inputenc}
\RequirePackage[leqno]{mathtools}
\RequirePackage{amsthm}
\RequirePackage{amssymb}
\RequirePackage[mathcal]{euler}
\RequirePackage{dsfont}
\RequirePackage{varioref}
\RequirePackage[numbers]{natbib}
\RequirePackage{enumerate}
\RequirePackage{version}
\RequirePackage{color}

\makeatletter
\@ifpackageloaded{hyperref}
{
	\def\mref#1#2{\hyperref[#2]{#1~\ref*{#2}}}
	\def\meqref#1#2{\hyperref[#2]{#1~(\ref*{#2})}}
	\def\mvref#1#2{\hyperref[#2]{#1~}\vref{#2}}
} {
	\def\mref#1#2{#1~\ref{#2}}
	\def\meqref#1#2{#1~\eqref{#2}}
	\def\mvref#1#2{#1~\vref{#2}}
	\newcommand{\hyperref}[2][1]{#2}
}
\makeatother

\def\bigdissum{\sum}
\def\dissum{+}

\def\indep{\perp}
\def\alg{\mathcal}
\def\A{\mathfrak{A}}
\def\genalg{\mathop{\sigma}}
\def\indalg#1{\langle #1 \rangle}
\def\restr#1{\mathord\mid_{#1}}
\def\charf#1{1_{#1}}
\def\m{\mathfrak{m}}
\def\F#1#2{\mathfrak{F}_{#1#2}}
\def\Fl#1#2#3{\mathfrak{F}_{#1#2}^{#3}}
\def\T#1#2{\mathcal{T}_{#1#2}} 
\def\expected{\text{\large $E$}}
\def\variance{\text{\large $D^2$}}
\def\symdiff{\vartriangle}
\def\tendsto{\longrightarrow}

\def\Img{\mathop{\mathrm{Im}}}
\def\ssum_#1{\smashoperator{\sum_{#1}}}
\def\sbigdissum_#1{\smashoperator{\bigdissum_{#1}}}
\def\numberfield{\mathbf}
\def\varkappa{\kappa}
\def\min{\mathop{\mathrm{min}}\nolimits}

\newcommand{\tsum}{\mathop{\textstyle \sum }}%

\newcommand{\defem}[1]{\textbf{#1}}
\renewcommand{\emph}{\textit}
\newcommand{\Emph}{\textbf}
\newcommand{\dash}{\nobreakdash-\hspace{0pt}}

\DeclarePairedDelimiter\abs{\lvert}{\rvert}
\DeclarePairedDelimiter\norm{\lVert}{\rVert}

\def\numberfield{\mathds}

\newcommand{\R}{ \numberfield{R} }

\newcommand{\tim}{\cdot}

\newtheorem{thm}{Theorem}
\newtheorem{rmk}{Remark}
\newtheorem{crr}{Corollary}
\newtheorem{prn}{Proposition}
\newtheorem{lem}{Lemma}

\newtheorem{dfn}{Definition}

\newtheorem{exm}{Example}

\newtheorem*{darboux}{Darboux Property}

\def\new#1{{\color[rgb]{0.6,0.15, 0.05}#1}}
\def\mu{P}

\excludeversion{ver:old}
\includeversion{ver:new}

\begin{document}

	\begin{abstract}
		We provide, under proper assumptions, a
		description of \defem{additive partition entropies}.
		They are real functions~$I$ on the set of finite partitions
		that are additive on stochastically independent partitions
		in a given probability space.

		\new{Second version, 2012-02-22. This version looks closer into the notion of continuity. All changes, with the exception of small typografic ones, to the previous version are shown in this colour.}
	\end{abstract}

	\title{Additive entropies of partitions}

		\address{
			Faculty of Mathematics and Computer Science\\
			University of Łódź\\
			ul. Banacha 22, 90-238 Łódź \\
			Poland}
		\author{Adam Paszkiewicz}
			\email[A.~Paszkiewicz]{ktpis@@math.uni.lodz.pl}
		\author{Tomasz Sobieszek}
			\email[T.~Sobieszek]{sobieszek@@math.uni.lodz.pl}
			\urladdr{http://sobieszek.co.cc}

	\thanks{This paper is partially supported by Grant nr N N201 605840.}



\maketitle
	\section {Introduction}

\begin{ver:new}
The classical discrete theory of entropy concerns itself with real functions $H$ defined on the family of sequences
$(p_1,\ldots,p_n)$ such that $p_i \ge 0$ and $\tsum p_i = 1$. There, the most significant role of all entropies is played
by the \defem{Shannon entropy}, given by
$H(p_1,\ldots,p_n)= p_1\log\tfrac 1{p_1} + \ldots + p_n \log \tfrac 1{p_n}$. (In here, as throughout the
paper, we confine ourselves to base~$2$ logarithms, as dictated by information theory tradition). It is the only symmetric (i.e. independent of the order of $p_i$-s) continous function
of such sequences that is normalised by $H(1/2,1/2)=1$ and satisfies the following \defem{grouping axiom}\footnotemark
\begin{multline}
	H(\ a_1 p_1, \ldots, a_k p_1\ , \  b_1 p_2, \ldots, b_l p_2\ , \ \cdots\ , \ c_1 p_n, \ldots, c_m p_n\ ) =\notag\\	
	\begin{aligned}
	&H (p_1, p_2, \cdots, p_n) \\
	&+ p_1 H(a_1, \ldots, a_k) + p_2 H(b_1,\ldots,b_l) + \cdots + p_n H(c_1,\ldots,c_m).
	\end{aligned}
\end{multline}
\footnotetext{This result which is a little better than Shannon's own set of axioms (see~\cite{S}) is a version
of Fadeev's axioms of entropy, c.f. \cite{F}. The grouping axiom is often called \defem{strong-additivity}. In terms of conditional entropy of random variables it can be succinctly expressed as the so-called \defem{chain rule} $H(X,Y) = H(X) + H(Y|X)$, (see \cite{CT}, Theorem 2.2.1). The shape of the grouping axiom, can lead us to think about the differences of entropies between two partitions, one coarser than the other. Indeed, a simple, but a very neat reformulation of the grouping axiom in these terms as a certain linearity property can be found in~\cite{BFL}}
Arguably, the grouping axiom can be singled out as the most important and almost defining property of
Shannon entropy. Even so, a noticeable part of Information Theory has been played by either some modifications or weakened statements of the grouping axiom. Among them there is the important and much weaker
property, called \defem{additivity},
satisfied not just by Shannon entropy but also by the earlier concept called Hartley entropy (see \mref {Example} {exm:numbering} in Section~2, c.f. \cite{H}) and Rényi entropy  (see \mref {Example} {exm:R} in Section~2, c.f. \cite{R}):
\begin{equation}	\label{eqn:entropyadditive}
H(p\otimes q) = H(p) + H(q),
\end{equation}
where
\begin{align*}
	&p=(p_1,\ldots,p_n), q=(q_1,\ldots,q_m),\\
	&p\otimes q = (p_1q_1,\ldots, p_1q_m, p_2q_1,\ldots, p_2q_m, \ldots, p_n q_1, \ldots, p_n q_m).
\end{align*}
Every symmetric $H$ that satisfies this equation is called an \defem{additive entropy}.
Hartley and Rényi entropies proved themselves useful and applicable in many fields.\footnote{Rényi entropy,
for instance, is used in random search~\cite{R2}, coding theory~\cite{Ca2}, cryptography~\cite{BBCM}, and differential geometric aspects of statistics~\cite{Ca}.}
Additivity can be found in some axiomatisations of Shannon entropy, one notable
instance that also gives an axiomatisation of Hartley entropy is given in~\cite{AFN}.

There are quite a few other entropies of sequences and their properties. For their detailed
exposition, see~\cite{AD}. For a modern survey of characterisations of Shannon entropy (among other
things), see~\cite{Cs}.

Now, the natural setting in which entropies appear in most applications is not on a family of sequences
but on a space of events. Indeed, a finite partition can be regarded as a representation of the information carried by experimentally collected data (see e.g. \cite{A}). It seems therefore only fitting and indeed often necessary to introduce the concepts
of entropy that take events into account. This has lead to the development of 'mixed theory of information' by
Aczél, Daróczy and others. This theory is a research of various so-called \defem{inset entropies} and was introduced in a series of papers, beginning with~\cite{AD2}.

Consider a ring $\alg B$ of subsets of~$\Omega$. Before, we were considering a family of sequences of positive numbers
$(p_1,\ldots,p_n)$ that sum up to~$1$. Riding roughshod over certain details,\footnotemark\,  we now consider a family $\alg G$ of sequences of pairs $(A_i,p_i)$ such that $A_i$-s
are elements of $\alg B$ which make up a partition, and $p_i$-s, as previously, are
positive and add up to~$1$. On such a family we consider functions $I:\alg G \to \R$.
\footnotetext{There are actually several different flavours of inset entropy.
For instance, it is sometimes assumed that the ring $\alg B$ is an algebra (contains $\Omega$)
and we consider partitions of $\Omega$, at other times we consider partitions of all sets in $\alg B$.
We can consider positive $p_i$-s or just nonnegative. The sets can be nonempty or not neccessarily.}

Various conditions analogous to symmetry, strong-additivity, additivity, and so on lead on to different inset entropies~$I$. For instance we can
define \defem{additive inset entropy} by demanding a version of symmetry 
and a version of additivity
\begin{align*}
		I\begin{pmatrix}	A_i\cap B_j		\\ 
					p_i q_j	\end{pmatrix}
		= I\begin{pmatrix}	A_i		\\ 
					p_i	\end{pmatrix} 
		+ I\begin{pmatrix}	 B_j		\\ 
					q_j	\end{pmatrix}
\end{align*}

Inset entropies have recently found application in considerations involving utility function in gambling, see~\cite{NLM}.

We propose a related, yet different, approach.
\end{ver:new}
We are given a nonatomic probability space $(\Omega, \alg \Sigma, P)$.
\begin{ver:old}
Take any algebra~$\alg A$ generated by a partition of~$\Omega$ into disjoint events
$A_1, \dotsc, A_n \in \alg \Sigma$, $A_i \cap A_j = \emptyset$. By the \defem{Shannon information function}
of this algebra we shall mean a mapping~$L_{\alg A}$ defined by
\[
	L_{\alg A} (\omega) := \log \frac 1{P(A_i)},
		\qquad \text{when } \omega \in A_i.
\]

According to information theory tradition we confine ourselves to base~$2$ logarithms.

\end{ver:old}Let $\A$ be the family of all finite subalgebras of
$\sigma$--algebra~$\alg \Sigma$, or in other words algebras generated by finite partitions
$\Omega = A_1 \cup \dotsb \cup A_n$, $A_i \cap A_j = \emptyset$.
\begin{ver:old}
A significant part of the well-developed mathematical
information theory is played by the axiom systems of \defem{Shannon entropy}, that is
a function $L_s\colon \A \to \R_+$ given by
\[
	L_s(\alg A) := \ssum_{1 \le i \le n}
		P(A_i) \log \frac1{P(A_i)} = \expected L_{\alg A},
\]
where $\alg A$ is generated by an $\{A_1, \dotsc, A_n\}$ partition of $\Omega$;
and $\expected$ denotes the expeted value, here the integral
$\int_\Omega L_{\alg A}\, dP$.
These axiom sets were derived by Khinchin, Fadeev, Urbanik (c.f. \cite {Ch1},
\cite {Ch2}, \cite {F}, \cite {IU}).
A detailed description can be found in~\cite{GT}.

A reccurring thought of the axiom systems of entropy is to require that
\end{ver:old}
\begin{ver:new}
Consider the following version of additivity
\end{ver:new}
\begin{equation}		\label{eqn:intromain}
	I\left( \sigma (\alg A \cup \alg B) \right)= I(\alg A) + I(\alg B),
\end{equation}
for any independent finite algebras~$\alg A$ and~$\alg B$.

It is the aim of the current paper to examine the family of functions $I\colon \A \to \R$
satisfying \meqref {condition} {eqn:intromain}, which we shall call
\defem {additive partition entropies}.\footnote{The term 'partition entropy'  has been introduced earlier in the context of partitions of a finite set, see e.g. \cite{SJ}. Additive entropy of partitions has been considered in~\cite{BN}.} It turns out that this research leads to a simple and
effective description.
\begin{ver:old}

	These considerations may be motivated by the fact that in many applications a
	finite partition of $\Omega$ is regarded as
	a representation of the information carried by experimentally collected data (see e.g. \cite{A}).
	Thus investigating more general concepts of information than
	$L_s$ seems to be well grounded.

\end{ver:old}

Of course, every additive entropy $H$ can be viewed as an additive partition entropy
$H_P$, $(H_P) (\sigma(A_1,\ldots,A_n)) = H(P(A_1), \ldots,P(A_n))$ that depends only on probabilities of atoms.

But there are other additive partition entropies. In fact, given $x\in \Omega$ consider the function
$L_x:\A \to \R$, 
\[
	L_x(\alg A) = \log \tfrac 1 {P(A_x)},\quad\text{where $A_x$ is the atom of $\alg A$ that contains $x$}
\]
Naturally, this function satisfies \meqref {equation} {eqn:intromain}.

Here is the main result of the paper:
``Let $I$ be a continuous additive partition entropy. There exist a continuous entropy $H$, and a
countably-additive set function~$\m$ absolutely continuous
with respect to probability~$P$ such that $I$ is a sum of two continous additive partition entropies
\begin{align*}
	I &= H_P \ + \ \int_\Omega \!\!\! L_x (\cdot) \m(dx),	\quad\text{that is}\\
	I(\alg A) &= H(P(A_1), \ldots,P(A_n)) + \sum_{i=1}^{n} \m(A_i)\log\tfrac 1{P(A_i)}
\end{align*}
Moreover we can assume that $\m(\Omega) = 0$, in this case such a decomposition is unique.''

Let us mention that a somewhat related result, however with much stronger assumptions, and in the setting of additive inset entropy has already appeared in~\cite{E}, see also~\cite{KS}.

The above result shall be presented as a corollary in~Section~5, \mref {Theorem}
{thm:probmain}. Sections 2--4 contain a proof of a slightly more general \mref{Theorem} {thm:cts};
namely it is natural to replace the probability~$P$ with a somewhat more general
notion of a finitely-additive measure~$\mu$. A follow up paper~\cite{So} concerns a similar result in which there are no continuity assumptions.

The proof of Theorem~1 is longer than might be expected.
The major difficulty lies in the construction of component $\int_\Omega \! L_x (\cdot) \m(dx)$, i.e.
 in finding the measure $\m$. The construction, is made of two
parts — one is probabilistic in that it depends heavily on the notion of independence (point I below), the other  involves algebraic manipulation of measures (points II and III below).

Theorem~1 is a corollary to the following results:
\begin{enumerate}[I]
	\item (cf. \mref{Propositions} {prn:key} and \ref {prn:delta}, Section 3)
	If $I$ is a continuous (see \mref {definition} {dfn:abscont} below) additive partition entropy
	then for any events $V$, $W$ with $P(V)=P(W)$ we can define the number
	$\Delta (V,W)$ in such a way that the following conditions are satisfied
	\begin{enumerate}
		\item
			Whenever there is a partition $(A_1,\ldots,A_n)$ with $V\subset A_1$,
			$W\subset A_2$, $P(A_2)/P(A_1)=\lambda$ we have
			\begin{align*}
				\Delta(V,W) \log\lambda &=
				I(A_1 \symdiff (V\cup W), A_2\symdiff (V\cup W),\ldots, A_n)
				\\&\quad - I(A_1,A_2,\ldots, A_n)
			\end{align*}
		\item
			If $P(U)=P(V)=P(W)$ then
			\[
				\Delta(U,W) = \Delta(U,V) + \Delta(V,W).
			\]
		\item
			For $V = V_1 \cup V_2$, $W = W_1 \cup W_2$,
			$V_1\cap V_2=\emptyset$, $W_1\cap W_2=\emptyset$,
			$P(V_i) = P(W_i)$, $i = 1,2$ we have
			\[
				\Delta(V,W) = \Delta(V_1,W_1) + \Delta(V_2,W_2).
			\]
		\item
			$\Delta(\cdot, \cdot)$ is continuous in metric $d(V,W)=P(V\symdiff W)$.
	\end{enumerate}
	\item (cf. \mref{Lemma} {lem:keycts}, Section 3)
		If $\Delta(V,W)$ defined for $P(V)=P(W)$ satisfies (b)–(d) then there is a
		unique finitely-additive measure $\m$, absolutely continous with respect to $P$
		such that $\m(\Omega)=0$ and
		\[
			\Delta(V,W)=\m(W)-\m(V).
		\]
	\item (cf. proof of Theorem 1, Section 4)
	Assume that for some continous additive partition entropy $I$, and
	measure $\m \ll P$ we have
		\[\begin{split}
			[\m(W)-\m(V)]\log\lambda&=
				I(A_1 \symdiff (V\cup W), A_2\symdiff (V\cup W),\ldots, A_n)\\
				&\quad - I(A_1,A_2,\ldots, A_n),
		\end{split}\]
	for $P(V)=P(W)$, and a partition $(A_1,\ldots,A_n)$ with $V\subset A_1$,
	$W\subset A_2$, $P(A_2)/P(A_1)=\lambda$. It turns out that
	$\widetilde I (\alg A):= I(\alg A) -\sum_{i=1}^{n} \m(A_i)\log\tfrac 1{P(A_i)}(\alg A)$ is a continous additive partition entropy
	such that $\widetilde I (\alg A)  = \widetilde I (\alg B)$ for any $\alg A$ and $\alg B$
	generated by $A_1, \dotsc, A_n$ and $B_1, \dotsc, B_n$ such that
	$P(A_i) = P(B_i)$, $1 \le i \le n$.
\end{enumerate}


All propositions say something about either continuous or general additive partition entropies
and they all follow easily from either continous-case \mref {Theorem} {thm:cts} or an analogical general result in~\cite{So}. The lemmas, however, are
of an `independent' nature and some of them, especially Lemmas 7–9, might find themselves
useful somewhere else.
\begin{ver:old}
Additive partition entropies $L_\m$ can find application in coding theory.
It turns out that the value $L_\m(\alg A)$ is often close to the minimal price
of sending a coded message about $\alg A$, when ${d\m \over dP} = g$, and
$g$ is a nonnegative bounded function that expresses the price of sending a bit
of information.
The details are to be expected in the upcoming paper \cite {P}.
\end{ver:old}

	\section {Basic notions and notation}

We shall denote the characteristic function of a set~$A$ by~$\charf{A}$.
We shall write $A = \bigdissum A_i$ or $A = A_1 \dissum\dotsb\dissum A_n\,$,
whenever we have $A=\bigcup_{1\le i \le n} A_i$ and the sets $A_i$ are pairwise disjoint.

Fix a space   $(\Omega, \alg{F}, \mu)$with a finitely-additive probability measure~$\mu$
defined on some algebra~$\alg F$ of subsets of~$\Omega$. We shall consider the family
$\A = \A( \alg F )$ of all finite subalgebras of the algebra~$\alg F$.
From now on, we shall hold on to the following assumption regarding measure~$\mu$
		\begin{darboux}		\label{rmk:assumption}
	For any set $A \in \alg F$ and any $0 < \theta < \mu(A)$ there is
	$B \in \alg F$ such that $B \subset A$ and $\mu(B) = \theta$.
	\end{darboux}

Let us remark that in case  $(\Omega, \alg F, \mu)$ is a usual probability space, the Darboux
property is satisfied if and only if the space is nonatomic.

Now, going back to the finitely-additive case, we have a naturally-defined notion of an integral of an
$\R \cup \{ +\infty \}$-valued simple function. We shall assume that $0 \tim (+\infty) = 0$.

For any $K\in \alg F$ such that $\mu(K)>0$, we shall consider a truncated conditional
probability space $(K, \alg F\restr{K}, \mu\restr{K})$. This means that
$A \in \alg F\restr{K}$  is equivalent to $A \in \alg F \wedge A \subset K$ and that
\[
	\mu\restr{K}(A) = \frac{\mu(A)}{\mu(K)}.
\]
Undeniably, any such ``subspace'' satisfies our
\hyperref[rmk:assumption]{Darboux Property}.

By $\A\restr K$ we shall understand $\A( \alg F\restr K )$, i.e. the
family of all finite subalgebras of~$\alg F\restr K$. For any
$\alg{A} \in \A$, by $\alg{A}\restr K \in \A\restr K$ we shall denote
$\{ A \cap K  \colon  A \in \alg{A} \}$.
For every finite family $\alg G \subset \alg F$, by
$\genalg(\alg G) \in \A$ we shall understand the algebra generated by~$\alg G$,
that is the smallest algebra containing~$\alg G$.

In the sequel we shell often refer to algebras generated by partitions. This is
why the following special notation might prove convenient. We shall write
\[
	\alg A = \indalg{A_1, \dotsc, A_n},
\]
if $\alg A = \genalg( \{A_1, \ldots, A_n\})$ and
$A_1 \dissum \dotsb \dissum A_n = \Omega$,
$A_i \neq \emptyset$, $A_i \in \alg F$.

Clearly, every algebra $\alg A \in \A$ is of shape
$\indalg{A_1, \dotsc, A_n}$. In fact, we have
\[
	\{ A_i : 1 \le i \le n \} = \min_\subset( \alg{A} )
\]
where $\min_\subset( \alg A )$ denotes the set of minimal elements of
$\alg{A} \setminus \{\emptyset\}$ partially ordered by the relation~$\subset$
(the set of \defem {atoms} of algebra~$\alg A$).

Notice that, for $\alg{A} = \indalg{A_i}$
and~$\alg{B} = \indalg{B_j}$ we have
$\genalg( \alg{A} \cup \alg{B} ) = \indalg{ C_1, \dotsc, C_s }$, where
$\{ C_1, \dotsc, C_s \} := \{ A_i \cap B_j : A_i \cap B_j \neq \emptyset,
1 \le i \le n, 1 \le j \le m \}$.

For any $\alg{A}, \alg{B} \in \A$ we shall write  $\alg{A} \indep \alg{B}$
whenever these algebras are independent, that is when $\mu(A\cap B) = \mu(A)\mu(B)$
for all $A \in \alg A$, $B \in \alg B$. One reason for using this symbol is given by the following
observation: ``Algebras $\alg A$ and~$\alg B$ are independent
if and only if $\int fg\,d\mu = 0$ for any simple functions $f$ and~$g$, which are
measurable with respect to the corresponding algebras $\alg A$ and~$\alg B$,
and which have their expected values equal to $0$, $\int f\,d\mu = \int g\,d\mu = 0$''.
Under the reprepresentations
$\alg{A} = \indalg{ A_1, \dotsc, A_n }$ and $\alg{B} = \indalg{ B_1,\dotsc, B_m}$,
the independence $\alg{A} \indep \alg{B}$ is equivalent to having
$\mu( A_i \cap B_j ) = \mu(A_i) \mu(B_j)$ for all $i,j$.
We shall also write
\[
	\alg{C} = \alg{A} \tim \alg{B},
\]
if $\alg{A} \indep \alg{B}$ and $\alg{C} = \genalg( \alg{A} \cup \alg{B} )$.

This paper is devoted to exploring functions of the following kind:
		\begin{dfn}
	The mapping $I\colon \A \to \R$ is called \defem{an additive partition entropy}
	if it satisfies
	\[
		I( \alg{A} \tim \alg{B} ) = I( \alg{A} ) + I( \alg{B} ),
	\]
	that is if for any $\alg A$, $\alg B \in \A$ such that
	$\alg A \indep \alg B$ we get
	$I( \genalg(\alg A \cup \alg B) ) =  I(\alg{A}) + I(\alg{B})$. \end{dfn}

Let us state several exampes that will play a role in the general description of
additive partition entropies. Before we do that, however, consider the following
crucial function $L:\A \to \R^\Omega$, which to a given algebra
$\alg A=\indalg{A_1,\dotsc,A_n}\in \A$ assigns a simple function
\[
	L(\alg A):=\ssum_{1\le i\le n}
	\:\left(\!  \log \frac 1{\mu(A_i)}  \!\right) \charf{A_i}.
\]
Evidently, the operator~$L$ solves our equation
$L( \alg{A} \tim \alg{B} ) = L( \alg{A} ) + L( \alg{B} )$. This leads to the followng
		\begin{exm}		\label{exm:H}
	Take an arbitrary finitely-additive measure~$\m\colon\alg F \to \R$,
	vanishing on sets of $\mu$-measure zero. We obtain an
	additive partition entropy
	\[
		L_\m(\alg A)  :=  \int\!  L(\alg A)  \,d\m  = 
		\ssum_{1\le i\le n}  \m(A_i) \log \frac 1{\mu(A_i)}.
	\] \end{exm}

A special case $\m = \mu$, gives the Shannon entropy of~$\alg A$
\[
	L_\mu(\alg A) = \expected L(\alg A) = 
									\ssum_{1\le i\le n}  \mu(A_i) \log \frac 1{\mu(A_i)}.
\]
By taking variance instead of expectation, we arrive at

		\begin{exm}		\label{exm:V}
	For any algebra $\alg A \in \A$ set
	\[
		V( \alg A ) := \variance \bigl[ L( \alg A ) \! \bigr] =
		\int \! \left[ L(\alg A) - \expected L(\alg A) \right]^2 d\mu.
	\]
	Since $V(\alg A)$ is the variance of the random variable~$L(\alg A)$
	with respect to probability~$\mu$, and since for any independent
	$\alg A, \alg B \in \A$ the random variable
	$L( \alg A\tim \alg B ) = L(\alg A) + L(\alg B)$ is the sum of
	independent random variables $L(\alg A)$ and~$L(\alg B)$, it follows that
	the mapping $V$ is an additive partition entropy. 
	\end{exm}

	Shannon entropy and \mref{Example}{exm:V} could be generalised further to any cumulant.
	In fact, consider the cumulant generating function (cgf) of $L$
	\[
		t\mapsto \log \big(\expected  \exp (tL)\big) = \log \big(\ssum_{1\le i \le n}\mu(A_i)^{1-t}\big).
	\]
	Recall that the cgf of a sum of independent random variables is equal to the sum of cgf's of the
	respective variables. Therefore, any 'linear functional on the cgf' is an example of an
	additive partition entropy. This can be the cumulants i.e. the coefficients of the power series
	representation of cgf of $L$. It is also convenient to consider the values of cgf
	(say at $t=1-\alpha$), up to a constant factor they are the

		\begin{exm}		\label{exm:R}
	\defem{Rényi entropies} of order $\alpha	\neq 1$, namely
	\[
		R_\alpha(\alg A) : = 
		\tfrac 1{1-\alpha} \log \big(\ssum_{1\le i \le n}\mu(A_i)^\alpha\big).
	\]
	As explained above, or seen directly, these are additive partition entropies. In fact, this is the most widely
	known class of examples of an additive entropy. (These entropies were introduced in~\cite{R}.)
	\end{exm}

The  role of the leading coefficient $1/(1-\alpha)$ is to ensure that the Rényi entropy tend to Shannon entropy as $\alpha$ tends to $1$, (this follows, for instance, from a simple use of de l'Hospitals rule, see also~\cite{AD}). The case $\alpha = 0$ is somewhat special, too. In this case
	
		\begin{exm}		\label{exm:numbering}
	$R_0$ takes the shape
	\[
		\alg A \mapsto \log\,\#\{A_i:\mu(A_i)>0\},
	\]
	and is known as \defem{Hartley entropy}.(It was defined in~\cite{H}.) \end{exm}

\new{Our final pair of examples have an important use  in the research of extreme cases
	of complexity theory, (see \cite {TWW}).
		\begin{exm}		\label{exm:minmax}
	They are the minimum, and the maximum of the simple function $L(\alg A)$:
	\begin{align*}
		L_\mathrm{min}(\indalg{A_1, \dotsc, A_n})
		&:= \mathrm{min}_{\substack{1\le i\le n \\
				\mu(A_i)\neq 0}} \log \tfrac 1{\mu(A_i)},\\
		L_\mathrm{max}(\indalg{A_1, \dotsc, A_n})
		&:= \mathrm{max}_{\substack{1\le i\le n \\
				\mu(A_i)\neq 0}}  \log \tfrac 1{\mu(A_i)}.
	\end{align*}	\end{exm}}

		\begin{dfn}
	We say that algebras $\alg A$ and $\alg B$ \defem{have the same
	measures of atoms}, if we have 
	\[
		\alg A = \indalg{A_1,\dotsc,A_n}, \qquad
		\alg B = \indalg{B_1,\dotsc,B_n}
	\]
	with $\mu (A_i) = \mu (B_i)$, for $1 \le i \le n$. Also,
	an additive partition entropy~$I$ will be said to
	\defem{depend solely on the measures of atoms} if $I$ is
	constant on each family of algebras that have the same measures of atoms.\end{dfn}
	
This is the same as saying that there is a 'classical' additive entropy~$H$, such that
$I=H_P$, i.e. $I(\alg A) = H(\mu(A_1), \ldots, \mu(A_n))$.

		\begin{rmk}		\label{rmk:type}
	The additive partition entropies of~\mref {Examples} {exm:V}--\ref {exm:minmax}
	depend solely on the measures of atoms. In cases when $\m$ is absolutely continuous
	with respect to~$\mu$, (see \mvref {Definition} {dfn:abscont})
	the additive partition entropy~$L_\m$ from \mref {Example} {exm:H}
	depends solely on the measures of atoms exactly when $\m = \alpha \mu$, for
	some~$\alpha \in \R$, \new{i.e. when it is a multiple of Shannon's entropy.}	\end{rmk}

Indeed, if $L_\m$ depends solely on the measures of atoms then
there is a function $f \colon \R \to \R$ such that for any $A \in \alg F$ we have
$\m(A) = f(\mu(A))$. (In fact, if
$\alg A = \indalg{A, \Omega \setminus A}$,
$\alg B = \indalg{B, \Omega \setminus B}$,
$\mu(A) = \mu(B) =: a \neq 1/2$ then
we put $x:=\m(A)$, $y:=\m(B)$ and get
\begin{equation}		\label{eqn:mcheck}
\begin{aligned}
	L_\m(\alg A) &= \left( \log \frac 1a \right) x +
				\left( \log \frac 1{1-a} \right) (1-x) =\\
	L_\m(\alg B) &= \left( \log \frac 1a \right) y +
				\left( \log \frac 1{1-a} \right) (1-y);
\end{aligned}
\end{equation}
from which $x=y$. If $a= 1/2$ break $A$ and $B$ into two pieces of equal measures.)
The function~$f$ must be additive and continuous at zero, thus linear.

Incidentally, observe that the mapping~$\m \mapsto L_\m$ is injective. This
follows in nearly the same way as equations~\eqref {eqn:mcheck} above.

	\section {Some introductory statements}

From now on, when speaking of an \defem{algebra} or
writing $\alg{A}$, $\alg{B}$,  etc. we shall mean an algebra in~$\A$, or in other words an
algebra generated by a partition of $\Omega$ into a finite number of
measurable sets.

For any algebras $\alg A$, $\alg B$, and $\alg K = \indalg{K_1,\dotsc,K_n}$ we shall write
$\alg A \indep_{\alg K} \alg B$ if for every $K_i$ such that $\mu(K_i) > 0$ the algebra
$\alg A \restr {K_i}$ is independent with $\alg B \restr {K_i}$.
		\begin{lem}\label{lem:separation}
	Let $\alg K = \indalg{K_1,\dotsc,K_n} \subset \alg A$ and $\mu(K_i) >0$, $1\le i \le n$.
	Then the following conditions are equivalent
	\begin{enumerate}
		\item $\alg A \indep \alg B$,
		\item  $\alg B \indep \alg K, \alg A \indep_{\alg K} \alg B$.
	\end{enumerate}
	If these conditions are satisfied then
	\[
		\alg C = \alg A \tim \alg B\quad\text{is equivalent to}\quad
		K_i \in \alg C,\ \alg C \restr {K_i} = \alg A \restr {K_i} \tim \alg B \restr {K_i},\ 1 \le i \le n.
	\] \end{lem}
		\begin{proof}
\begin{ver:old}
	If $\alg A \indep \alg B$ then $\alg B \indep \alg K$ and for any $A \in \alg A$, $B \in \alg B$
	we have $\mu(A \cap K_i) \, \mu(B) = \mu(A \cap B \cap K_i)$. This gives
	\[
		\frac {\mu(A \cap B \cap K_i)} {\mu(K_i)} =
		\frac {\mu(A \cap K_i) \, \mu(B)} {\mu(K_i)} =
		\frac {\mu(A \cap K_i)}{\mu(K_i)} \tim
					\frac {\mu(B \cap K_i)} {\mu(K_i)}.
	\]	
	On the other hand, whenever $A \in \alg A$, $B \in \alg B$ condition 2 gives
	\[
		\mu( A\cap B \cap K_i ) =
		\frac { \mu(A\cap K_i) \, \mu(B \cap K_i) } {\mu(K_i)} 	
		 = \mu(A \cap K_i) \, \mu(B)
							\qquad (1 \le i \le n).
	\]
	Then
	\[\begin{split}
		\mu(A \cap B) = \sum_i \mu(A \cap B \cap K_i)
		= \sum_i \mu(A \cap K_i) \, \mu(B) = \mu(A) \, \mu(B).
	\end{split}\]
\end{ver:old}
\begin{ver:new}
	We leave to the reader the proof that $1.$ and~$2.$ are equivalent.

\end{ver:new}
	To show the second equivalence we need the following for $\alg K \subset \alg C$
	\begin{equation}		\label{eqn:tmp1}
		\alg C = \genalg(\alg A \cup \alg B) \iff
		\alg C\restr{K_i} = \genalg( \alg A\restr{K_i} \cup
				\alg B\restr{K_i} ) \qquad (1 \le i \le n).
	\end{equation}
	Since $\alg K \subset \alg C$ we have
	\begin{equation}\label{eqn:tmp2}
		\min_\subset( \alg C ) = \bigdissum_i
					\min_\subset( \alg C\restr{K_i} ).
	\end{equation}
	Similarly from $\alg K \subset \genalg(\alg A \cup \alg B)$ we obtain
	\begin{equation}		\label{eqn:tmp3}
		\min_\subset( \genalg(\alg A \cup \alg B) )
		= \bigdissum_i \min_\subset( \genalg(\alg A \cup \alg B)
								\restr{K_i} )
		= \bigdissum_i \min_\subset( \genalg(\alg A\restr{K_i} \cup
							\alg B\restr{K_i}) ).
	\end{equation}

	From \eqref{eqn:tmp2} and \eqref{eqn:tmp3}
	we obtain \eqref{eqn:tmp1}. \end{proof}

		\begin{lem}		\label{lem:anytype}
	For any algebra $\alg A$ and any numbers $c_1, \dotsc, c_k \ge 0$ such
	that $\sum_{1 \le j \le k} c_j = 1$ there is an algebra
	$\alg C = \indalg{C_1, \dotsc, C_k}$ satisfying
	$\alg C \indep \alg A$, $\mu(C_i) = c_i$. \end{lem}

		\begin{proof}
	This follows easily from the
	\hyperref[rmk:assumption]{Darboux property}.\end{proof}

		\begin{crr}		\label{crr:indep}
	For any algebras $\alg A_1, \dotsc, \alg A_n$ and any nonnegative numbers
	$c_1, \dotsc, c_k$ such that $\sum_{1 \le j \le k} c_j = 1$
	there is an algebra $\alg C = \indalg{C_1, \dotsc, C_k}$ satisfying
	$\alg C \indep \alg A_i$, $1 \le i \le n$, $\mu(C_i) = c_i$. \end{crr}

		\begin{proof}
	Consider $\alg A = \genalg(\alg A_1, \ldots, \alg A_n)$.\end{proof}

		\begin{rmk}		\label{rmk:indep}
	We will often use this Collorary in the following way.
	For any algebras $\alg A$ and $\alg B$ having the same measures of atoms
	there is an algebra $\alg C$ with the same measures of atoms such that
	\[
		\alg{C} \indep \alg{A}, \quad \alg{C} \indep \alg{B}.
	\] \end{rmk}

Fix any additive partition entropy~$I$.

		\begin{prn}		\label{prn:zero}
	If for some set~$Z$, $\mu(Z) = 0$ we have
	\[
		\alg A \restr {\Omega \setminus Z} =
		\alg B \restr {\Omega \setminus Z},
	\]
	then
	\[
		I(\alg{A}) = I(\alg{B}).
	\] \end{prn}
	
		\begin{proof}
	We can easily find representations
	$\alg{A} = \indalg{ A_1, \dotsc, A_n, \dotsc, A_s }$ and
	$\alg{B} = \indalg{ B_1, \dotsc, B_n, \dotsc, B_t }$
	for which the following equalities hold
	\[
		\mu( A_i \symdiff B_i ) = 0  \quad  (1 \le i \le n),
		\qquad
		\mu(A_{n+k}) = 0, \quad \mu(B_{n+l}) = 0 \quad (k,l \ge 1).
	\]

	Observe that any algebra $\alg{C}$ generated by sets
	of measure~$0$ is independent with any other algebra from $\A$. The algebra
	$\alg{C} := \genalg( \{ A_i \symdiff B_i, A_{n+k}, B_{n+l} \} )$
	is of such a shape and $\alg{A} \tim \alg{C} = \alg{B} \tim \alg{C}$.
	\end{proof}

		\begin{prn}		\label{prn:equal}
	Consider algebras $\alg{A} = \indalg{ A_1, \dotsc, A_k, A_{k+1}, \dotsc, A_n }$
	and $\alg{A'} = \indalg{ A'_1, \dotsc, A'_k, A_{k+1}, \dotsc, A_n }$ with
	\[
		\mu(A_1)= \ldots = \mu(A_k)  \quad  = \quad
		\mu(A'_1) = \ldots =\mu(A'_k).
	\]
	Then
	\[
		I(\alg{A}) = I(\alg{A'}).
	\] \end{prn}

		\begin{proof}
	Put $K := A_1 \dissum\dotsb\dissum A_k$. Since the case $\mu(K) = 0$
	is handled by the previous Proposition we assume  $\mu(K) > 0$. By the same Proposition we shall assume $K = \Omega$ whenever $\mu(K) = 1$. \mref{Corollary}{crr:indep} (see \mref{Remark}{rmk:indep}) when applied
	to~$\A\restr K$ lets us suppose that
	$\alg{A}\restr K \indep \alg{A'}\restr K$.

 	We shall find $\alg{B}$ such that $\alg{A} \tim \alg{B} =
	\alg{A'} \tim \alg{B}$. Firstly, define an algebra
	$\alg{B}_K = \indalg{ B^{(1)}_K, \dotsc, B^{(k)}_K }\in \A\restr K$
	by
	\[
		B^{(i)}_K \::=\: \sbigdissum_{\substack{q-p \,\equiv\, i \;
				(\mathrm{mod}\, k)\\
		                                1 \le p,q \le k}}
		A_p \cap A'_q  \qquad  (1 \le i \le k).
	\]
	Then $\alg{A}\restr{K} \tim \alg{B}_K =
	\alg{A}\restr{K} \tim \alg{A'} \restr{K}  = 
	\alg{A'}\restr{K} \tim \alg{B}_K$.

	If $K \neq \Omega$ then \mref{Lemma}{lem:anytype}
	allows us to find an algebra
	\[
		B_{\Omega \setminus K} =
		\indalg{ B^{(1)}_{\Omega \setminus K}, \dotsc,
		B^{(k)}_{\Omega \setminus K} } \in \A\restr{\Omega \setminus K}
	\]
	that has the same measures of atoms as $\alg{B}_K$
	and satisfies
	$B_{ \Omega \setminus K } \indep \alg{A}\restr{ \Omega \setminus K }$.
	Consider
	\[
		\alg{B} :=
		\indalg{ B^{(1)}_K \cup B^{(1)}_{\Omega \setminus K}, \dotsc,
		B^{(k)}_K \cup B^{(k)}_{\Omega \setminus K} } \in \A.
	\]
	By \mref{Lemma}{lem:separation} we obtain the equality
	$\alg{A} \tim \alg{B} = \alg{A'} \tim \alg{B}$. \end{proof}

		\begin{rmk}		\label{rmk:H}
	As we have seen, all our examples can be made to rely on the function $L$ defined just before
	\mref {Example} {exm:H}. This is by no coincidence. \mref{Propositions}{prn:zero} and \ref{prn:equal}
	mean that if $L(\alg A)$ and $L(\alg B)$ differ on a set of
	measure~$0$ then $I( \alg A )  =  I( \alg B )$; in particular every
	additive partition entropy~$I$ factors through~$L$, i.e. $I = J \circ L$, where
	\[
		J( \phi + \psi ) = J(\phi) + J(\psi)
		\quad\text{with } \phi, \psi \in \Img(L)\text{ and $\phi$, $\psi$ independent.}
	\]
	Moreover it seems quite plausible that this condition is satisfied for $\phi, \psi, \phi+\psi \in \Img(L)$.
	In this paper, we don't pursue this approach any further. \end{rmk}

For any disjoint sets $V, W \in \alg F$
define a nonempty family~$\F V W$ of algebras~$\alg A$ which satisfy
$V \subset A_1$, $W \subset A_2$ for some representation
$\indalg{A_1,\dotsc,A_n}$ of algebra~$\alg A$. Also define an operation
$\T V W \colon \F V W \to \F V W$ by
\[
	\T V W \,\indalg{A_1,\dotsc,A_n} =
	\indalg{(A_1 \setminus V) \cup W, (A_2 \setminus W) \cup V,A_3,\dotsc,
							A_n},
\]
when $V \subset A_1$, $W \subset A_2$.

		\begin{lem}	\label{lem:diffandprod}
	Consider a pair of disjoint sets $V, W$ with
	$\mu(V) = \mu(W)$, and algebras $\alg A, \alg B \in \F VW$.
	Whenever
	\[
		\alg C = \alg A \tim \alg B,
	\]
	we also have
	\[
		\T VW \alg C = \T VW \alg A
					\tim \T VW \alg B.
	\] \end{lem}

		\begin{proof}
	We can set
	\[
		\alg A = \indalg{A_1,\dotsc,A_n}, \quad
		\alg B = \indalg{B_1,\dotsc,B_m}, \quad
		\alg C =\indalg{C_1,\dotsc,C_p},
	\]
	with  $V \subset A_1 \cap B_1 = C_1$, $W \subset A_2 \cap B_2 = C_2$.
	Put $A'_i := A_i \symdiff V \symdiff W$,
	$B'_j := B_j \symdiff V \symdiff W$ and
	$C'_k := C_k \symdiff V \symdiff W$ for $i,j,k = 1,2$, and  also 
	$A'_i = A_i$ , $B'_j = B_j$, $C'_k = C_k$ for $3 \le i \le n$,
	$3 \le j \le m$, $3 \le k \le p$.
	Then $\indalg{A'_1,\dotsc,A'_n} = \T VW \alg A$,
	$\indalg{B'_1,\dotsc,B'_m} = \T VW \alg B$,
	$\indalg{C'_1,\dotsc,C'_p} = \T VW \alg C$ and
	$\indalg{C'_1,\dotsc,C'_p} =
		\genalg(A'_1,\dotsc,A'_n,B'_1,\dotsc,B'_m)$.

	It suffices to show that
	\begin{equation}	\label{eqn:diffandprod1}
		\begin{gathered}
			\mu( A'_i ) = \mu( A_i), \qquad \mu(B'_j) = \mu(B_j), \\
			\mu( A'_i \cap B'_j ) = \mu( A_i \cap B_j).
		\end{gathered}
	\end{equation}
	When $i = j = 1, 2$ this follows from equalities $\mu(V) = \mu(W)$ and
	inclusions $V \subset A_1 \cap B_1$, $W \subset A_2 \cap B_2$. For all
	remaining pairs $(i,j)$, $i = 1, \dotsc, n$, $j = 1, \dotsc, m$ the last
	of~\meqref{equalities in}{eqn:diffandprod1} is also clear, and we even
	have $A'_i \cap B'_j = A_i \cap B_j$. \end{proof}

For any $\lambda > 0$ put
\[
	\varepsilon(\lambda) := \min \left( 1/(1+\lambda), 1/(1+\lambda^{-1})
								\right).
\]
This notation has the following sense --- if we divide~$\Omega$
into two sets having their quotient of measures equal to~$\lambda$
then $\varepsilon(\lambda)$ will be the measure of the smaller of them.

Whenever $V \cap W = \emptyset$ and $\lambda >0$ write
\[\begin{split}
	\Fl V W \lambda := \big\lbrace \alg A \in \A \colon &
	\tfrac{\mu(A_2)}{\mu(A_1)} = \lambda, V \subset A_1, W \subset A_2 \\
					&\text{ for some representation }
			\alg A = \indalg{A_1, \dotsc, A_n}\big\rbrace.
\end{split}\]

		\begin{lem}	\label{lem:epsilon}
	For any $\lambda>0$ we have what follows:

	\Emph{A.} For any pair of disjoint sets $V, W$
	such that $\mu(V), \mu(W) \le \varepsilon(\lambda)$ there is
	an algebra $\alg A = \indalg{A_1, A_2} \in \Fl V W  \lambda$.

	\Emph{B.} For any algebra $\alg{A} = \indalg{ A_1, \dotsc, A_n }$,
	$n \ge 2$ with $\mu(A_2) / \mu(A_1) = \lambda$, a number
	$\varkappa > 0$ and sets $V \subset A_1$, $W \subset A_2$ that satisfy
	\[
		\mu(V), \mu(W) \le \varepsilon(\varkappa)\varepsilon(\lambda)
		\mu(A_1 \dissum A_2)
	\]
	there is an algebra $\alg B = \indalg{B_1, B_2}$
	with the property that $\alg A \indep \alg B$, $\alg B \in \Fl V W  \varkappa$.

	\Emph{C.} For any disjoint $V, W$ and $\varkappa >0$ such that
	$\mu(V), \mu(W) \le \varepsilon(\varkappa)\varepsilon(\lambda)$,
	there exist algebras $\alg A = \indalg{A_1, A_2}$, $\alg B =
	\indalg{B_1, B_2}$ such that
	\[
		\alg A \in \Fl V W  \lambda, \quad
		\alg B \in \Fl V W  \varkappa, \quad\text{and}\quad
		\alg A \indep \alg B.
	\]

	\Emph{D.} If $\alg A \in \Fl V W \lambda$,
	$\alg B \in \Fl V W \varkappa$ and $\alg C = \alg A \tim \alg B$ then
	$\alg C \in \Fl V W {\lambda \varkappa}.$ \end{lem}

		\begin{proof}
	\Emph{A.} If $\mu(V), \mu(W) \le \varepsilon(\lambda)$ then
	by the \hyperref[rmk:assumption]{Darboux property} there is
	a partition $\Omega = A_1 \dissum A_2$ such that
	$\mu(A_1) = 1/(1+\lambda)$ and $\mu(A_2) = 1/(1+\lambda^{-1})$ with
	$V \subset A_1$, $W \subset A_2$. Then $\mu(A_2) / \mu(A_1) = \lambda$.

        \Emph{B.} By the assumptions $\mu\restr{A_1}(V),
	\mu\restr{A_2}(W) \le \varepsilon(\varkappa)$. According
	to~\Emph{A.} there exist sets $C_i^{(j)} \in \A \restr {A_i},$
	$1 \le i \le n$, $j = 1,2$ such that $A_i = C_i^{(1)} + C_i^{(2)}$,
	$\mu(C_i^{(2)})/\mu(C_i^{(1)})=\varkappa$, and $V \subset C_1^{(1)}$,
	$W \subset C_2^{(2)}$.
	We finish by writing
	\[
		B_j := C_1^{(j)} \dissum \dotsb \dissum C_n^{(j)}.
	\]

	\Emph{C.} follows from \Emph{A.} and~\Emph{B.}, whereas~\Emph{D.} is
	obvious. \end{proof}

The following lemma is utterly straightforward.

		\begin{lem}		\label{lem:taumor}
	Whenever $\alg A \in \Fl VW\lambda$, we have $\alg A \in \Fl  {V_i}{W_i}\lambda$, $1 \le i \le n$ and
	\[
	\T VW \alg A = \T {V_1}{W_1} \dotsm \T {V_n}{W_n} \alg A
	\]
	for any $V=V_1 \dissum \dotsb \dissum V_n$, $W=W_1 \dissum \dotsb \dissum W_n$.
	\end{lem}

		\begin{rmk}
	It follows from \mref{Proposition}{prn:equal} that if $\alg A \in \Fl VW1$ then
	\[
		I( \T V W  \alg A ) = I(\alg A).
	\]
	\end{rmk}

		\begin{prn}		\label{prn:same}
	Consider $\lambda > 0$ and a pair of disjoint sets $V, W$ with
	$\mu(V) = \mu(W)$.

	If $\alg A, \alg B \in \Fl V W  \lambda$ then
	\[
		I( \T V W  \alg A ) - I(\alg A) =
		I( \T V W  \alg B ) - I(\alg B).
	\] \end{prn}

		\begin{proof}
	Fix algebras $\alg A$ and $\alg B$.

	By \mref{Proposition}{prn:zero} we can assume that $\mu(V) > 0$;
	then for some $\varkappa>0$ we get
	$\genalg (\alg A \cup \alg B) \in \Fl V W \varkappa$.

	Suppose for the moment that we also have
	\[
		\mu(V) = \mu(W) \le \varepsilon(1/\lambda) \tim \varepsilon(\varkappa)
			\tim \mu(A_1 \cap B_1 \dissum A_2 \cap B_2),
	\]
	where $A_i, B_i$ are such atoms of algebras $\alg A$ and $\alg B$ that
	$V \subset A_1 \cap B_1$ and $W \subset A_2 \cap B_2$.

	With the use of \mref{Lemma}{lem:epsilon}B we find an algebra
	$\alg C = \indalg{ C_1, C_2 } \in \Fl V W {1/\lambda}$ such that
	$\genalg (\alg A \cup \alg B) \indep \alg C$. Then,
	using~\mref{Lemma}{lem:diffandprod}, \mref{Lemma}{lem:epsilon}D
	and the Remark above we get the equalities
	$I( \alg A \tim \alg C ) = I( \T V W  \alg A \tim \T V W  \alg C )$ and
	$I( \alg B \tim \alg C ) = I( \T V W  \alg B \tim \T V W  \alg C )$.
	We are done.
	
	In the general case divide the sets $V$ and $W$ into the same number of
	pieces of equal measure, the measure being bound by
	\[
		\varepsilon(1/\lambda) \tim \varepsilon(\varkappa)
				\tim \mu(A_1 \cap B_1 \dissum A_2 \cap B_2).
	\]
	Subsequently apply \mref{Lemma}{lem:taumor} for $\alg A, \alg B \in \Fl V W  \lambda$ and
	$\genalg (\alg A \cup \alg B) \in \Fl V W \varkappa$
	to the already derived instance of this Proposition where we have constraints
	on the size of $V$ and $W$.\end{proof}

		\begin{prn}\label{prn:key}
	For any sets $V, W$ with
	$\mu(V) = \mu(W)$ and any $\lambda > 0$ there is a unique
	$\Delta(V,W,\lambda)$ such that the following conditions are satisfied:
	\begin{enumerate}
		\item		\label{prop:deltadef}
			Whenever $V\cap W=\emptyset$ and there exists
			$\alg A\in \Fl VW\lambda$ we have
			\begin{equation}		\label{eqn:deltadef}
				\Delta(V,W,\lambda) =
				I(\T V W \alg A) - I(\alg A),
			\end{equation}
		\item		\label{prop:cycle}
			For any sets $U$, $V$ and $W$ of the same measure $\mu$
			we have
			\[
				\Delta(U,W,\lambda) = \Delta(U,V,\lambda) +
				\Delta(V,W,\lambda).
			\]
		\item		\label{prop:meslike}
			For $V = V_1 \dissum V_2$, $W = W_1 \dissum W_2$,
			$\mu(V_i) = \mu(W_i)$, $i = 1,2$ we have
			\[
				\Delta(V,W,\lambda) = \Delta(V_1,W_1,\lambda) +
				\Delta(V_2,W_2,\lambda).
			\]
		\item		\label{prop:loglike}
			For $\varkappa > 0$ we have
			\[
				\Delta(V,W,\kappa\lambda) = \Delta(V,W,\kappa) +
				\Delta(V,W,\lambda).
			\]
	\end{enumerate} \end{prn}

		\begin{proof}
	Notice at first that in the case when the family $\Fl V W \lambda$
	is nonempty, Proposition \ref {prn:same} proves that the right hand side of
	\meqref {equality} {eqn:deltadef} does not depend on the algebra
	$\alg A \in \Fl V W \lambda$, i.e. that \meqref {formula} {eqn:deltadef}
	defines the quantity $\Delta(V,W,\lambda)$ well.	
	Observe  also  that by \mref {Lemma} {lem:epsilon}A the family $\Fl V W \lambda$ is
	nonempty when the sets $V, W$ are disjoint and satisfy $\mu(V) = \mu(W) < \varepsilon(\lambda)$.

	The proof is made of three parts.
	
	\textbf{I.~} Assume that $\Fl V W \lambda \neq \emptyset$ and define for now
	$\Delta$ simply by~\eqref {eqn:deltadef}.
	In particular we assume that $V, W$ are disjoint.
	Using just defined $\Delta$'s we shall prove \mref{property} {prop:cycle} when
	$\mu(V) = \mu(W) < (1/2)\, \varepsilon(\lambda)$,
	\mref {property} {prop:meslike} assuming that
	$\mu(V) = \mu(W) < \varepsilon(\lambda)$ and
	\mref {property} {prop:loglike} if
	$\mu(V) = \mu(W) < \varepsilon(\varkappa) \varepsilon(\lambda)$.
	
	Indeed, under these assumptions we shall show
	\mref {property} {prop:cycle} by choosing an algebra $\alg A$
	in such a way that $U \in A_1$, $V \cup W \in A_2$, where $A_2$ is the atom of
                 $\alg A$ with~$\lambda$ times bigger $\mu$-measure than~$A_1$,
	and noting that
	\[
		\alg A \in \Fl U W \lambda,		\quad
		\alg A \in \Fl U V \lambda,		\quad
		\T U V \alg A \in \Fl V W \lambda,	\quad
		(\T V W \circ \T U V) \alg A = \T U W \alg A.
	\]
	Then
	\begin{align*}
		\Delta(U,W,\lambda)
		&	= I(\T U W \alg A) - I(\alg A)	\\
		&	= I( (\T V W \circ \T U V)\alg A) - I(\T U V \alg A)
			  +I(\T U V \alg A) - I(\alg A) \\
		&	= \Delta(V,W,\lambda) + \Delta(U,V,\lambda).
	\end{align*}

	In order to show \mref{property}{prop:meslike} we select
	$\alg A \in \Fl V W \lambda$ and obtain
	\begin{align*}
		\Delta(V,W,\lambda)
		&	= I(\T {V_2}{W_2} \circ \T {V_1}{W_1} \alg A)
			  - I(\alg A)	\\
		&	= I(\T {V_2}{W_2} \circ \T {V_1}{W_1} \alg A)
			  - I(\T {V_1}{W_1} \alg A) + I(\T {V_1}{W_1} \alg A)
			  - I(\alg A)	\\
		&	= \Delta(V_2,W_2,\lambda) + \Delta(V_1,W_1,\lambda).
	\end{align*} 

	It remains to prove \mref{property}{prop:loglike}. Observe that by
	\mref{Lemma}{lem:epsilon}C there exist $\alg A \in \Fl V W \lambda$,
	$\alg B \in \Fl V W \varkappa$ such that $\alg A \indep \alg B$.
	What is more, by \mref{Lemma}{lem:epsilon}D we have
	$\alg A \tim \alg B \in \Fl V W {\lambda \tim \varkappa}$;
	now using \mref{Lemma}{lem:diffandprod} we get
	\begin{align*}
		\Delta(V,W,\lambda \varkappa)
		&	= I(\T V W (\alg A \tim \alg B))
			  - I(\alg A \tim \alg B)	\\
		&	= I(\T V W \alg A) - I(\alg A)
			  + I(\T V W \alg B) - I(\alg B)	\\
		&	= \Delta(V,W,\lambda) + \Delta(V,W,\varkappa).
	\end{align*}

	Notice also that if $\mu(V) < \varepsilon(\lambda)$ then
	\[
		\Delta(V,W,\lambda) = - \Delta(W,V,\lambda),
	\]
	which we shall use in the next step of the proof.
	
	\textbf{II.~} We drop now the assumption that the sets
	$V$ and $W$ be disjoint, that is we consider the case when
	$\Fl {V\setminus W} {\,W\setminus V} \lambda \neq \emptyset$
	(e.g. when $\mu(V) = \mu(W) < \varepsilon(\lambda)$),
	and define $\Delta$ by
	\[
		\Delta(V, W, \lambda) :=
				\Delta(V\setminus W, W\setminus V, \lambda)
	\]

	By supposing, if neccessary, that we have $\mu(V) < 1/4$ in addition
	to the assumptions of part~\textbf{I}, we shall show the required
	\mref{properties}{prop:cycle} and \ref {prop:meslike}
	(\mref{property}{prop:loglike} is obvious),
	without assuming that $V$ and $W$ are disjoint.

	In order to to show these properties we note at first that for
	any~$U$, $V$ and $W$ having the required properties
	and for any~$X$ such that $\mu(X) =\mu(V)$ and
	$X \cap (U \cup V \cup W) = \emptyset$ we have
	\[
		\Delta(V, W, \lambda) =
				\Delta(V,X,\lambda) +\Delta(X,W, \lambda)
	\]
	Indeed set $V' := V \setminus W$, $W' := W \setminus V$,
	$A= V \cap W$ then  divide the set~$X$ into two parts --- $X'$~of measure
	$\mu(X') = \mu (V')$ and $B$ of measure $\mu(B) = \mu(A)$.
	Using step \textbf{I}., we obtain:
	\begin{align*}
		\Delta(V,W,\lambda)
		&	= \Delta(V',W',\lambda)	\\
		&	= \Delta(V',X',\lambda) + \Delta(X',W', \lambda) \\
		&	= \Delta(V',X',\lambda) +\Delta(A,B,\lambda)
			  + \Delta(X',W', \lambda)+\Delta(B,A,\lambda)	\\
		&	= \Delta(V,X,\lambda) +\Delta(X,W, \lambda).
	\end{align*}

	From the obtained equality we arrive at \mref{property}{prop:cycle}:
	\begin{align*}
		\Delta(U,W, \lambda)
		&	= \Delta(U,X, \lambda) + \Delta(X, W, \lambda)	\\
		&	= \Delta(U, X, \lambda) + \Delta(X, V, \lambda)
			  + \Delta(V, X, \lambda) + \Delta(X, W, \lambda) \\
		&	= \Delta(U, V, \lambda) + \Delta(V, W, \lambda),
	\end{align*}
	and also, by a division of the set $X$ into the parts $X_1$, $X_2$, with their
	measures equal to $\mu(V_1)$ and $\mu(V_2)$, respectively we arrive at
	\mref{property}{prop:meslike}:
	\begin{align*}
		\Delta(V,W,\lambda)
		&	= \Delta(V,X,\lambda) + \Delta(X,W,\lambda)	\\
		&	= \Delta(V_1,X_1,\lambda) + \Delta(V_2,X_2,\lambda)
			  + \Delta(X_1,W_1,\lambda) + \Delta(X_2,W_2,\lambda) \\
		&	= \Delta(V_1,W_1,\lambda) + \Delta(V_2,W_2,\lambda).
	\end{align*}
	
	\textbf{III.~} In the general case ($\mu(V) = \mu(W)$, $\lambda >
	0$) the number $\Delta(V, W, \lambda)$ can be uniquely defined
	by
	\[
		\Delta(V,W,\lambda) := \ssum_{1 \le i \le k}
		\Delta(V_i,W_i,\lambda),
	\]
	where $V = V_1 \dissum \dotsb \dissum V_k$, $W = W_1 \dissum \dotsb
	\dissum W_k$, $\mu(V_i) = \mu(W_i) < \varepsilon(\lambda)$; and where
	$\Delta(V_i,W_i,\lambda)$ is defined as in part~\textbf{II}.

	We shall show now that the number $\Delta(V,W,\lambda)$ is well-defined in this way .
	For any partitions $V = V_1 \dissum \dotsb \dissum V_k$,
	$W= W_1 \dissum \dotsb \dissum W_k$ and
	$V = V'_1 \dissum \dotsb \dissum V'_l$,
	$W= W'_1 \dissum \dotsb \dissum W'_l$ satisfying respectively
	$\mu(V_i) = \mu(W_i) < \varepsilon(\lambda)$, $1 \le i \le k$
	and $\mu(V'_j) = \mu(W'_j) < \varepsilon(\lambda)$, $1 \le j \le l$,
	we choose another pair of partitions
	$V = V''_1 \dissum \dotsb \dissum V''_m$,
	$W = W''_1 \dissum \dotsb \dissum W''_m$ in such a way that
	$\mu(V''_i) = \mu(W''_i) < \varepsilon(\lambda)$, and
	$\indalg{V''_1,\dotsc,V''_m}$ is 
	independent with $\indalg{V_1,\dotsc,V_k}$ and $\indalg{V'_1,\dotsc,V'_l}$
	in the space $(V,\alg F \restr V, \mu \restr V)$, and also that
	$\indalg{W''_1,\dotsc,W''_m}$ is 
	independent with $\indalg{W_1,\dotsc,W_k}$ and $\indalg{W'_1,\dotsc,W'_l}$
	in the space $(W,\alg F \restr W, \mu \restr W)$ . Then
	\begin{align*}
		\mu(V_i \cap V''_j)	&= \mu(W_i \cap W''_j)	\\
		\mu(V'_i \cap V''_j)	&= \mu(W'_i \cap W''_j).
	\end{align*}
	Hence
	\[
		\ssum_{1 \le i \le k} \Delta(V_i, W_i, \lambda) =
		\ssum_{\substack{1 \le i \le k \\
				1 \le j \le m} }
			\Delta(V_i \cap V''_j, W_i \cap W''_j, \lambda) =
		\ssum_{1 \le j \le m} \Delta(V''_j, W''_j, \lambda).
	\]
	In the same way
	\[
		\ssum_{1 \le i \le l} \Delta(V'_i, W'_i, \lambda) =
		\ssum_{1 \le j \le m} \Delta(V''_j, W''_j, \lambda).
	\]

	It is easy to see that \mref {Properties} {prop:cycle}--\ref {prop:loglike} are
	satisfied in all their generality. \end{proof}

\begin{rmk}\label{rmk:zero}
It is worthy to note that \mref {Properties} {prop:cycle} and \ref {prop:meslike}
entail formulas
\begin{gather*}
	\Delta(V,W,\lambda) = - \Delta(W,V,\lambda),\\
	\Delta(V,W,\lambda) = \Delta(V \setminus W, W \setminus V, \lambda);
\end{gather*}
moreover, according to \mref {Proposition} {prn:zero},
if $\mu(V \symdiff V') = \mu(W \symdiff W') = 0$, then
\[
	\Delta(V,W,\lambda) = \Delta(V',W',\lambda).
\]
\end{rmk}

	\section {The main result}

In this section we provide a description of continous additive partition entropies. However, no notion of continuity has been developed as yet. We would like our notion to be as weak as possible,  with many continous partition entropies. Yet, at the same time we want partition entropies $H_P$ to be continuous exactly when
the corresponding classical enropies $H$ are continuous.
\new{Here, a subtle distinction should be made. There are two natural definitions of continuity for a classical entropy, depending on whether we consider nonnegative or just positive probabilities. The same 
phenomenon arises in the context of partition entropies.}

We shall consider a pseudometric~$d_{\alg F}$ in algebra $\alg F$
defined by
\[
	d_{\alg F} (A,B) = \mu(A \symdiff B).
\]
\new{By analogy with the classical case we want to consider two topologies  in the family $\A$ both introduced as the richest topology so that, in the first case, the mappings ``similar to the following ones''
\[
	\{A \in \alg F : 0 \le \mu(A) \le 1 \}  \ni\; A
	\longmapsto \indalg{ A, \Omega \setminus A } \;\in \A
\]
and, in the second case, the following ones}
\[
	\{A \in \alg F : 0 < \mu(A) < 1 \}  \ni\; A
	\longmapsto \indalg{ A, \Omega \setminus A } \;\in \A
\]
are continuous.
\new{To do so we define the closed-domain topology in $\A$ by the following pseudometric:
\[
	d(\alg A, \alg B)
		:= \inf \, \{ \mu(Z) \colon \alg A \restr {\Omega \setminus Z}
		= \alg B \restr {\Omega \setminus Z} \},
\]
and the stronger open-domain topology by the following one}
\[
	D(\alg A, \alg B)
		:= \inf \, \{ \mu(Z) \colon \alg A \restr {\Omega \setminus Z}
		= \alg B \restr {\Omega \setminus Z} \}
			+ \abs{ N(\alg A) - N(\alg B) },
\]
where $N(\alg A)$ denotes the number of atoms of algebra~$\alg A$ with nonvanishing
measure.

\new{
		\begin{dfn}	\label{dfn:Icts}
	A function $I:\A \to \R$ is said to be \defem{closed-domain continuous} if it is continuous
	in metric~$d$ and \defem{open-domain continuous} or simply \defem{continuous} if it is
	continuous in metric~$D$.\end{dfn} }

\new{We consider closed-domain continuity to be too restrictive, with many entropies, like some $L_\m$ and Hartley entropy not being closed-domain continuous. Open-domain continuity on the other hand, plays nicely with our examples and our theory. To see that let us first make precise the following}

		\begin{dfn}		\label{dfn:abscont}
	A finitely-additive set function~$\m\colon\alg F \to \R$
	is said to be \defem{absolutely continuous} with respect to
	measure~$\mu$ $(\m \ll \mu)$ if it is continuous
	in pseudometric~$d_{\alg F}$, or equivalently if for any
	$\varepsilon > 0$ there is a
	$\delta > 0$ such that if we have $\mu(A) < \delta$ then we also have
	$\m(A) < \varepsilon$. \end{dfn}

Recall that in case of finitely-additive measures
the vanishing of $\m$ on sets of $\mu${\dash}measure~$0$ does not imply
the absolute continuity of~$\m$. \new{Now we have}

		\begin{rmk}		\label{rmk:examples}
	Additive partition entropy $L_\m$ from \mref{Example}{exm:H} is open-domain continuous,
	when $\m$ is absolutely continuous with respect to~$\mu$.
	(As it will follow from \mref {Theorem} {thm:cts} and \mref {Remark} {rmk:type}
	there are no other continuous additive partition entropies of this shape.) \end{rmk}

We shall show the continuity of $L_\m$ in case when $\m \ll \mu$.
Fix algebra $\alg A = \indalg{A_1, \dotsc, A_n}$ and $0 < \varepsilon < 1$.
There is a $0 < \delta < 1$ such that if
\[
	\mu(A_i \symdiff B) < \delta,\quad 1\le i\le n,
\] then
\[
	\abs{ \log \mu(A_i) - \log \mu(B) } < \varepsilon \quad\text{and}\quad
	\abs{ \m(A_i) -\m(B) } < \varepsilon.
\]
If now $d(\alg A, \alg B) < \delta < 1$ then the algebras $\alg A$ and $\alg B$
have the same number of atoms of nonzero measure. We can
assume that $\alg B = \indalg{B_1, \dotsc, B_n}$ and $\mu(A_i), \mu(B_i) > 0$,
and also that $\mu( A_i \symdiff B_i )  <  \delta$. Then
\begin{align*}
	\abs{ L_\m( \alg A ) - L_\m( \alg B ) }
	&=	\abs{ \sum \m(A_i) \log \tfrac 1{\mu(A_i)}
		- \sum \m(B_i) \log \tfrac 1{\mu(B_i)}	}	\\
	&\le	\sum \abs{ \m(A_i)
		\left( \log \tfrac 1{\mu(A_i)} - \log \tfrac 1{\mu(B_i)} \right)
			 }\ 	+\\
	&\qquad\qquad\qquad\qquad\qquad	 \sum \abs{ \big( \m(A_i) -\m(B_i) \big)
			\log \tfrac 1{\mu(B_i)} }	\\
	&<	\sum \abs{ \m(A_i) } \varepsilon +
		\sum \varepsilon \log \tfrac 1{\mu(B_i)}	\\
	&<	\varepsilon \sum \abs{ \m(A_i) }
		+ \varepsilon  \sum \left(\! \varepsilon
		+ \log \tfrac 1{\mu(A_i)} \!\right)	\\
	&<	\varepsilon \tim \mathrm{const}( \alg A ).	
\end{align*}

			\begin{rmk}	\label{rmk:seqentrcts}
	Given a 'classical' additive entropy $H$, the correspondig partition entropy $H_P$ is continuous if and only if $H$ is continous on the open domain,
	that is if each function
	\[
		H|_{\{(p_1,\ldots,p_n): p_i> 0, \sum p_i =1 \}}
	\]
	 is continuous.  \new{$H_P$ is closed-domain continuous iff $H$ is continuous everywhere.} It follows that, entropies of Examples
	\ref{exm:V}--\ref{exm:minmax} are continuous.\footnotemark\end{rmk}
\footnotetext{Please note that Hartley entropy, \new{and $L_\mathrm{min}$} are not closed-domain continuous.}

\new{It turns out that  in the derivation of our main theorem (\mref{Theorem} {thm:cts}) we can do with weaker concepts than that of open-domain continuity. In fact, we do not need the continuity on the full family ~$\A$. It will be sufficient to assume that $I$ is continuous on the family of algebras with $2$ atoms.
Let for that matter $\A_2$ denote the family
\[
	\A_2:=\{\indalg{A,B}: 0<\mu(A)<1\}.
\]}

\new{
		\begin{dfn}	\label{dfn:IctsA2}
	We say that $I$ is  \defem{continuous on $\A_2$} if the restriction $I_{\A_2}$ of $I$ to ${\A_2}$
	is continuous, 
	that is if for any sequence of sets $A, A_1, A_2, \ldots \in \alg F$ with $0<\mu(A)<1$ such that
	$\mu(A \symdiff A_n) \tendsto 0$ we have
	\[
		I(\indalg{A_n, \Omega \setminus A_n}) \tendsto I(\indalg{A, \Omega \setminus A}).
	\] \end{dfn}
}

		\begin{prn}		\label{prn:delta}
	If the additive partition entropy~$I$ is continuous \new{on $\A_2$}  then for any sets
	$V, W$ such that $\mu(V) = \mu(W)$ there is
	a $\Delta(V,W)$ satisfying
	\[
		\Delta(V,W,\lambda) = \Delta(V,W) \tim \log \lambda \qquad
							(\lambda > 0).
	\]
	What is more $\Delta$ has the following properties:
	\begin{enumerate}
		\item		\label{prop:cycle2}
			For any sets $U$, $V$ and $W$ of the same measure
			\[
				\Delta(U,W) = \Delta(U,V) + \Delta(V,W).
			\]
		\item		\label{prop:meslike2}
			For $V = V_1 \dissum V_2$, $W = W_1 \dissum W_2$,
			$\mu(V_i) = \mu(W_i)$, $i = 1,2$ we have
			\[
				\Delta(V,W) = \Delta(V_1,W_1) + \Delta(V_2,W_2).
			\]
		\item		\label{prop:cont}
			$\Delta(\cdot, \cdot)$ is \new{uniformly} continuous
			(in the topology induced from the one in $\alg F$.)
	\end{enumerate} \end{prn}

		\begin{proof}
	By $\Delta(V,W,\lambda) =
	\Delta(V \setminus W, W \setminus V, \lambda)$
	and \mref {property} {prop:meslike} of \mref {Proposition} {prn:key},
	in order to prove the first part of Proposition, we can assume that
	the sets $V$ and $W$ are disjoint and that their measures are smaller than~$1/2$.

	At first, we shall show that the mapping
	$\lambda \mapsto \Delta(V,W,\lambda)$ is continuous at $\lambda=1$.
	Indeed, consider the following open neighbourhood of the point~$1 \in \R$:
	$G := \{\lambda : \mu(V) < \varepsilon(\lambda) \}$.
	Consider also any sequence $(\lambda_i)_{i\ge 1}$ such that
	$\lambda_i \in G$ and $\lambda_i \tendsto 1$.
	By the \hyperref[rmk:assumption] {Darboux property}
	(see also \mref {Lemma} {lem:epsilon}A) we shall find
	algebras $\alg A_i = \indalg{A^i_1, A^i_2} \in \Fl V W {\lambda_i}$ and
	$\alg A = \indalg{A_1, A_2} \in \Fl V W 1$ such that
	$\mu({A^i_1 \symdiff A_1}) \tendsto 0$ when $i \tendsto \infty$.
	Then, from the continuity of~$I$ we get
	\begin{align*}
		I(\alg A_i) &\tendsto I(\alg A),	\\
		I(\T V W \alg A_i) &\tendsto I(\T V W \alg A),
	\end{align*}
	and consequently
	$\Delta(V,W,\lambda_i) \tendsto \Delta(V,W,1)$.

	Next by \mref {property} {prop:loglike} of \mref {Proposition} {prn:key} and the fact that $\Delta(V, W, \cdot)$ is continuous
	at $\lambda = 1$ we see that there is a constant $a \in \R$
	with
	\[
		\Delta(V, W, \lambda) = a \, \log \lambda \qquad (\lambda > 0).
	\]
	Put $\Delta(V, W) := a$.
	
	\mref {Property} {prop:cycle2} and \mref {property} {prop:meslike2}
	follow from the corresponding properties in \mref {Proposition} {prn:key}.

	As it follows  from \mref {property} {prop:meslike2}, \new{to get
	uniform continuity we only need to show that $\Delta$ is continuous} at
	$V = W = \emptyset$; i.e. it suffices to show that for any
	$\varepsilon > 0$ we can find $\delta >0$ such that if
	$\mu(V) = \mu(W) < \delta$ then $\abs {\Delta(V,W)} < \varepsilon$.

	Fix $\varepsilon >0$. Using the continuity of the function~$I$ at the point
	$\indalg{K_1,K_2}$, for any fixed partition $\mu(K_1) = 1/3$,
	$\mu(K_2) = 2/3$ of the space $\Omega$ we shall find $\delta > 0$ such that
	for any algebras $\indalg{L_1,L_2}$ which satisfy
	$\mu( K_1 \symdiff L_1 )  <  4\delta$ (and $\mu(L_1) = 1/3$)
	we have
	$\abs{ I(\indalg{K_1,K_2}) - I(\indalg{L_1,L_2}) }  <  \varepsilon/2$.
	
	Now, let the sets $V, W$ satisfy $\mu(V) = \mu(W) < \delta$.
	Put $V' = V \setminus W$, $W' = W \setminus V$.
	We can assume that 
	$\delta < 1/6$; then we can find sets $V''$ and $W''$ such that
	$V', W', V'', W''$ are disjoint and which satisfy the following equalities
	\begin{gather*}
		\mu(V' \cap K_1) = \mu(V'' \cap K_2),  \qquad
		\mu(V' \cap K_2) = \mu(V'' \cap K_1),  \\
		\mu(W' \cap K_1) = \mu(W'' \cap K_2),  \qquad
		\mu(W' \cap K_2) = \mu(W'' \cap K_1).
	\end{gather*}
	Set
	\begin{gather*}
		L_1 := (K_1 \cup V' \cup V'') \setminus (W' \cup W''),	\qquad
		L_2 := \Omega \setminus L_1,	\\
		L'_1 := L_1 \symdiff V' \symdiff W',	\qquad
		L'_2 :=\Omega \setminus L'_1.
	\end{gather*}
	Then
	\begin{gather*}
		\T {V'} {W'} \,\indalg{L_1, L_2} = \indalg{L'_1, L'_2},	\qquad
		\mu(L_1) = \mu(L'_1) = 1/3,	\\
		\mu( K_1 \symdiff L_1 )  =2\mu(V') <  2\delta,	\qquad
		\mu( K_1 \symdiff L'_1 ) \leq 4\mu(V') < 4\delta.
	\end{gather*}
	From this we get
	\[
		\abs{ \Delta(V,W) } = \abs{ \Delta(V',W',2) }
		= \abs{ I(\indalg{L'_1,L'_2}) - I(\indalg{L_1,L_2}) }
		< \varepsilon.
	\] \end{proof}
%

	Recall that $\mu$ is defined on the $\sigma$-algebra~$\alg F$.
	Let $\alg R$ be the family of all sets in $\alg F$ whose measure is rational.
		\begin{lem}	\label{lem:keyrational}
	Let $\Delta(V,W) \in \R$ be defined for any pair of sets
	$V,W \in \alg R$ of the same measure~$\mu$ and let it satisfy the following conditions
	\begin{enumerate}
		\item
			For any sets $U$, $V$ and $W$ of the same measure
			\[
				\Delta(U,W) = \Delta(U,V) + \Delta(V,W).
			\]
		\item
			For $V = V_1 \dissum V_2$, $W = W_1 \dissum W_2$,
			$\mu(V_i) = \mu(W_i)$, $i = 1,2$ we have
			\[
				\Delta(V,W) = \Delta(V_1,W_1) + \Delta(V_2,W_2).
			\]
	\end{enumerate}
	There is a unique finitely-additive set function~$\m\colon \alg R \to \R$ with
	$\m(\Omega)=0$ such that for any sets
	$V$ and $W$ with $V\cap W\in \alg R$ (i.e. belonging to the same
	algebra $\alg G \subset \alg R$) and satisfying $\mu(V) = \mu(W)$ we have
		\begin{equation}		\label{eqn:deltam}
		\Delta(V, W) = \m(W) - \m(V).
	\end{equation}
	What is more, for any $A \in \alg R$ we have
	\begin{equation}			\label {eqn:abscont}
		\abs{m(A)} \le \sup \abs{ \Delta(\cdot,A) }.
	\end{equation}

\end{lem}

		\begin{proof}
	Let $A_1 \dissum \dotsb \dissum A_n = \Omega$ with $\mu(A_i) = 1/n$.
	Write
	\begin{equation}		\label{eqn:mdef}
		\m(A_1 \dissum \dotsb \dissum A_k) := \frac 1 {\binom n k}
			\smashoperator[r]{
				\sum_{\substack{ \{ i_1, \dotsc, i_k\}\\
					\subset \{ 1, \dotsc, n \}	}  }
					}
			\;\Delta( A_{i_1} \dissum \dotsb \dissum A_{i_k}, 
				A_1 \dissum \dotsb \dissum A_k),
	\end{equation}
	(where by writing $\{ i_1, \dotsc, i_k\}$ we assume that these numbers
	are all distinct). Using \mref{property}{prop:meslike2} of \mref
	{Proposition} {prn:delta} we obtain
	\begin{equation}			\label{eqn:madd}
		\m(A_1 \dissum \dotsb \dissum A_k)
		= \frac 1 {k!\binom n k}\quad\ \;
			\smashoperator{
				\sum_{\substack{
						(i_1, \dotsc, i_k):	\\
						\{ i_1, \dotsc, i_k\}
						\subset \{ 1, \dotsc, n \}   \\
						   1 \le j \le k
				}	}
			}
			\;\Delta(A_{i_j}, A_j) = \tfrac 1 n
				\,\ssum_{\substack{
							1 \le i \le n	\\
							1 \le j \le k
				}		}
			\Delta(A_i, A_j).
	\end{equation}
	
	We show first that $\m$ is well defined. Indeed, if
	$\alg D = \indalg{D_1,\dotsc,D_{ns}}$, $\mu(D_p) = 1/(ns)$ and
	$\alg A=\indalg{A_1, \dotsc, A_n} \subset \alg D$ say $A_i = D_{(i-1)s+1} \dissum \dotsb
	\dissum D_{is}$, $1\le i\le n$ then
	\[
		\ssum_{\substack{(i-1)s+1 \le p \le is \\
				 (j-1)s+1 \le q \le js } } \Delta(D_p,D_q) =
		s\Delta(A_i,A_j).
	\]
	Hence
	\begin{equation}	\label{eqn:mAD}
		\tfrac 1 n
			\,\ssum_{\substack{ 1 \le i \le n	\\
						 1 \le j \le k	}  }
		\Delta(A_i, A_j) = \tfrac 1 {ns}
			\,\ssum_{\substack{ 1 \le p \le ns	\\
						 1 \le q \le ks	}  }
		\Delta(D_p, D_q).
	\end{equation}
	If we have now any other partition $B_1 \dissum \dotsb \dissum B_m
	= \Omega$, $\mu(B_j) = 1/m$ such that
	$A_1 \dissum \dotsb \dissum A_k = B_1 \dissum \dotsb \dissum B_l$ then
	there is a partition $C_1 \dissum \dotsb \dissum C_s$, $\mu(C_v) = 1/s$
	such that
	$C_1 \dissum \dotsb \dissum C_r = A_1 \dissum \dotsb \dissum A_k$
	and
	\begin{align*}
		\indalg{C_1, \dotsc, C_r}
			&\indep_{A_1 \dissum \dotsb \dissum A_k}
			\indalg{A_1, \dotsc, A_k}, \indalg{B_1, \dotsc, B_l}  \\
		\indalg{C_{r+1},\dotsc,C_s}
			&\indep_{A_{k+1} \dissum \dotsb \dissum A_n}
			\indalg{A_{k+1},\dotsc,A_n},\indalg{B_{l+1},\dotsc,B_m}.
	\end{align*}
	(Naturally, the symbol $\indep_K$ denotes the conditional independence
	of algebras in $\A \restr K$, with respect to the truncated conditional measure $\mu \restr K =
	\mu(\cdot)/\mu(K)$.) By the rationality of measure $\mu(D)$, where
	$D$ is any atom of the algebra $\alg A \tim \alg C$, we can find
	an algebra $\alg D \supset \alg A, \alg C$, all atoms of which
	$D_1, \dotsc, D_d$ are of the same measure. In particular, $d$ is a 
	multiple of both $n$ and $s$; moreover we can assume
	that for some~$e$ we have
	$D_1 \dissum \dotsb \dissum D_e = A_1 \dissum \dotsb \dissum A_k$.
	Observe now that \meqref {formula} {eqn:mAD} gives
	\[
		\tfrac 1 n
			\!\ssum_{\substack{ 1 \le i \le n	\\
					 1 \le j \le k	}  }
		\Delta(A_i, A_j) = \tfrac 1 d
			\!\ssum_{\substack{ 1 \le v \le d	\\
					 1 \le w \le e	}  }
		\Delta(D_v,D_w) = \tfrac 1 s
			\!\ssum_{\substack{ 1 \le h \le s	\\
					 1 \le k \le r	}  }
		\Delta(C_h,C_k).
	\]
	Similarly
	\[
		\tfrac 1 m
			\!\ssum_{\substack{ 1 \le v \le m	\\
					 1 \le w \le l	}  }
		\Delta(B_v,B_w) = \tfrac 1 s
			\!\ssum_{\substack{ 1 \le h \le s	\\
					 1 \le k \le r	}  }
		\Delta(C_h,C_k).
	\]
	This proves that $\m$ is well defined.
	
	The additivity of the mapping $\m$ follows from \meqref
	{equality} {eqn:madd}.
	What is more we have
	$\Delta(V,W) = \m(W) - \m(V)$ for $V \cap W \in \alg R$. Indeed, if $V=A_1 \dissum \dotsb \dissum A_k$ and
	$W=A_{\sigma(1)} \dissum \dotsb \dissum A_{\sigma(k)}$ where $\mu(A_1) = \ldots = \mu(A_n)$ and $\sigma(i) \in \{1,\ldots,n\}$ then
	\begin{align*}
		\Delta(V,W)
		&= \frac 1 {\binom n k}
			\smashoperator[r]{
				\sum_{\substack{ \{ i_1, \dotsc, i_k\}\\
					\subset \{ 1, \dotsc, n \}	}  }
					}
			\;\Delta( A_{1} \dissum \dotsb \dissum A_{k}, 
				A_{\sigma(1)} \dissum \dotsb \dissum A_{\sigma(k)})\\
		&=\frac 1 {\binom n k}
			\smashoperator[r]{
				\sum_{\substack{ \{ i_1, \dotsc, i_k\}\\
					\subset \{ 1, \dotsc, n \}	}  }
					}
			\;\Delta( A_{i_1} \dissum \dotsb \dissum A_{i_k}, 
				A_{\sigma(1)} \dissum \dotsb \dissum A_{\sigma(k)}) \ -\\ 
		&\qquad\qquad\qquad\qquad	\Delta( A_{i_1} \dissum \dotsb \dissum A_{i_k}, 
				A_{1} \dissum \dotsb \dissum A_{k}) \\
		&= \m(W) - \m(V).
	\end{align*} Definition \eqref{eqn:mdef} now gives $\m(\Omega)=0$. Conversely, if $\m$ satisfies $\m(\Omega)=0$
	and \eqref{eqn:deltam} it must satisfy \eqref{eqn:mdef} and so $\m$ is unique. 
	Inequality \eqref{eqn:abscont} follows from \meqref {definition} {eqn:mdef}. \end{proof}

		\begin{lem}	\label{lem:keycts}
	Let $\Delta(V,W) \in \R$ be defined for any pair of sets
	$V,W \in \alg F$ of the same measure~$\mu$ and let it satisfy \mref
	{conditions} {prop:cycle2}--\ref{prop:cont} of \mref {Proposition}
	{prn:delta}. There is a unique finitely-additive set function
	$\m\colon \alg F \to \R$, absolutely continuous with respect to measure~$\mu$
	such that $\m(\Omega) =0$ and such that for any sets $V$ and $W$ of the same
	measure~$\mu$ we have an equality
	\[
		\Delta(V, W) = \m(W) - \m(V).
	\]
	What is more, for any $A \in \alg F$ we have
	$\abs{m(A)} \le \sup \abs{ \Delta(\cdot,A) }$.	\end{lem}

		\begin{proof}
	We will extend $\m$ from \mref {Lemma}{lem:keyrational} onto the whole algebra~$\alg F$.

	\new{With the help of uniform continuity of $\Delta$ it is easy to show that for any measurable
	set $A$ the supremum of $\abs{ \Delta(\cdot,A) }$ is finite and that the
	mapping $A \mapsto \sup \abs{ \Delta(\cdot,A) }$ is continuous.}

	Let now $\alg R \ni A_i \tendsto A \in \alg F$, i.e. $\mu(A \symdiff A_i) \tendsto 0$.
	Then we have $\lim_{i,j \tendsto +\infty} \mu(A_i \setminus A_j) = 0$,
	and then by \eqref {eqn:abscont}
	$\lim_{i,j \tendsto +\infty} \m(A_i \setminus A_j) = 0$. Hence
	$\lim_{i,j \tendsto +\infty} \abs{\m(A_i) - \m(A_j)} = 0$, which means that
	the sequence $m(A_i)$ converges. Put $\m(A) := \lim \m(A_i)$. Clearly
	$\m(A)$ is well defined; (indeed, by taking any other sequence
	$B_i \tendsto A$, it suffices to consider the combined sequence
	$A_1, B_1, A_2, B_2, \dotsc \tendsto A$  thereby getting the convergence of
	$\m(A_1), \m(B_1), \m(A_2), \dotsc$).  We now get \meqref {inequality}
	{eqn:abscont} for any $A \in \alg F$; in particular it signifies the absolute
	 continuity of~$\m$ with respect to~$\mu$.
	The additivity of~$\m$, the equality $\Delta(V,W) = \m(W) -\m(V)$ and the uniqueness of $\m$
	follow from analogical properties obtained previously for rational values of~$\mu$.
	\end{proof}

Notice that the converse of the above lemma is trivially valid --- having~$\m$ as in
the thesis of Lemma, $\Delta$ satisfies \mref {conditions} {prop:cycle2}--\ref{prop:cont}
of \mref {Proposition} {prn:delta}.

We shall go over now to the crucial theorem of this paper.

		\begin{thm}		\label{thm:cts}
	Every additive partition entropy~$I$, continuous \new{on $\A_2$} has a unique decomposition
	into additive partition entropies:
	\[
		I = \widetilde I + L_\m,
	\]
	where $\widetilde I$ is an additive partition entropy depending solely on the measures
	of atoms, whereas~$L_\m$ is a continuous additive partition entropy from
	\mref{Example}{exm:H} for some finitely-additive set function~$\m$ absolutely
	continuous with respect to~$\mu$ such that $\m(\Omega) =0$.
	If $I$ is continuous, then so is $\widetilde I$.\end{thm}

		\begin{proof}
	Let $\m$ will be as in \mref {Lemma} {lem:keycts} for
	$\Delta$ of \mref {Proposition} {prn:delta}.
	Observe at first that the additive partition entropy $\widetilde I$ defined by
	\[
		\widetilde I (\alg A) := I(\alg A) - L_\m(\alg A) =
		 I(\alg A) - \left[ \m(A_1) \log\tfrac 1{\mu(A_1)} + \dotsb
		+ \m(A_n) \log\tfrac 1{\mu(A_n)} \right],
	\]
	where $\alg A = \indalg{A_1, \dotsc, A_n}$, is invariant in the following sense
	\[
		\widetilde I (\T V W \alg A) = \widetilde I (\alg A) \qquad
			\text{for $V, W$ such that $\alg A \in \F V W$,
							$\mu(V) = \mu(W)$.}
	\]
	Indeed, after assuming that $V \subset A_1$, $W \subset A_2$, setting
	\[
		B_i := A_i \symdiff V \symdiff W \in \T V W \alg A,\quad i = 1,2
	\]
	and using the equalities
	\[
		\Delta(V,W) = \m(W) - \m(V) =
		\m(B_1) - \m(A_1) = \m(A_2) -\m(B_2)
	\]
	we obtain
	\begin{align*}
		I( \T V W \alg A ) - I( \alg A )
		&= \Delta(V, W) \tim \log \tfrac{\mu(A_2)}{\mu(A_1)}	\\
		&= \left( \m(B_1) - \m(A_1) \right) \log\tfrac 1{\mu(A_1)}
                   - \left( \m(A_2) - \m(B_2) \right) \log\tfrac 1{\mu(A_2)} \\
		&=  \left[ \begin{array}{l}
			\phantom{+\ } \m(B_1) \log\tfrac 1{\mu(B_1)}	\\
			+\ \m(B_2) \log\tfrac 1{\mu(B_2)} \end{array} \right]
		   - \left[ \begin{array}{l}
			\phantom{+\ } \m(A_1) \log\tfrac 1{\mu(A_1)}	\\
			+\ \m(A_2) \log\tfrac 1{\mu(A_2)} \end{array} \right] \\
		&= L_\m( \T V W \alg A ) - L_\m( \alg A ).
	\end{align*} 

	It remains to show that by a sequence of operations of shape $\T V W$,
	like above, we can go from any algebra~$\alg A$ to any other
	algebra~$\alg B$, with the same measures of atoms as~$\alg A$.
	Let $\alg A = \indalg{A_1, \dotsc, A_n}$
	and $\alg B = \indalg{B_1, \dotsc, B_n}$, $\mu(A_i) = \mu(B_i)$. According
	to~\mref{Corollary}{crr:indep} (see \mref {Remark} {rmk:indep})
	we may assume that $\alg A \indep \alg B$.
	Then the sets $K_{i,j} := A_i \cap B_j$ are disjoint and
	$\mu(K_{i,j}) = \mu(K_{j,i})$; moreover for $i \neq j$ the sets
	$K_{i,j}$ and $K_{j,i}$ are contained in distinct atoms
	$A_i$, $A_j$ of the partition $\{ A_1, \dotsc, A_n \}$.
	It is easy to see that we have the following equality:
	\[
		\alg B = \T {K_{n-1,n}}{K_{n,n-1}} \dotsi
		\T {K_{1,3}}{K_{3,1}} \T {K_{1,2}}{K_{2,1}} \alg A.
	\]

	The uniqueness of the decomposition follows from \mvref {Remark} {rmk:type}.
	The continuity of the additive partition entropy $L_\m$  is described in \mref {Remark}
	{rmk:examples}. \end{proof}

\begin{ver:old}

We shall interpret now \meqref {inequality} {eqn:abscont} as an estimate of the measure~$\m$
with respect to some modulus of continuity of~$I$.

For any continuous function~$\Delta$ satisfying \mref {conditions}
{prop:cycle2}--\ref{prop:cont} of \mref {Proposition} {prn:delta}
the continuity of~$\Delta$ implies the finiteness and the continuity at $0$ of
the following modulus of continuity
\[
	\omega_\Delta(A) := \sup \abs{ \Delta(\cdot,A) }.
\]
It has  the following properties, entailed by the properties of~$\Delta$:
\begin{align}
	\omega_\Delta(\Omega \setminus A) &= \omega_\Delta(A)
				& &(A \in \alg F), \notag \\
	\omega_\Delta(A) &\le 2 \omega_\Delta(B) 	& &(\mu(A) = \mu(B)).
		\label{eqn:omegaself}
\end{align}

Observe that whenever $\Delta$ is derived from the continuous additive partition entropy~$I$
its modulus~$\omega_\Delta$ can also be considered as a kind of a modulus of continuity for
the function~$I$.
Indeed if $\mu(A) < 1/2$, we have
\begin{equation}		\label{eqn:omegaeq}
	\omega_\Delta(A)  \ =\  \smashoperator{
			\sup_{ \substack{\mu(V) = \mu(W) > 0 \\
			V \subset A, \alg A \in \F VW} } }
			\qquad \tfrac {\abs{ I(\T VW \alg A) - I(\alg A) }}
			{\log \left(\!  \tfrac {\mu(A_W)}{\mu(A_V)}  \!\right)},
\end{equation}
where $A_V$, $A_W$ are the atoms of the algebra $\alg A$
which contain respectively $V$ and $W$.
What is more, for any $\alg A \in \F VW$ we have
\begin{equation}		\label{eqn:deq}
	d(\T VW \alg A,\alg A) =\mu(V) + \mu(W).
\end{equation}

We also can define the following modulus of continuity of the function~$I$ along the
family of algebras with the same measures of atoms
\[
	\widetilde \omega_I  (\alg A,x) :=
		\smashoperator{
		\sup_{
		\substack{
			d(\alg A,\alg B) \le x \\
			\alg A \simeq \alg B
		}}}
		\abs{ I(\alg B) - I(\alg A) },
\]
where $\alg A \simeq \alg B$ denotes that these algebras have the same measures of atoms.
Put
\[
	c(\alg A) := \sup_{A, B}
		\tfrac 1 { \log \left( \tfrac {\mu(A)} {\mu(B)} \right) },
\]
where  $A$, $B$ are distinct atoms of $\alg A$ of nonzero measure~$\mu$.
Then, using \eqref {eqn:omegaself}, \eqref {eqn:omegaeq}
and~\eqref {eqn:deq} we get the following strenghening of the condition of the absolute
continuity of the measure~$\m$:

		\begin{prn}
	For any set $A \in \alg F$ and any algebra $\alg A$
	containing two atoms with their measures equal at least to~$\mu(A)$
	we have
	\[
		\abs{ \m(A) } \le \omega_\Delta(A)
		\le 2 c(\alg A) \widetilde \omega_I(\alg A, 2\mu(A)).
	\]	\end{prn}

Define the following pseudonorms in the space of $\alg F$-measurable
simple functions.
\[
	\norm {\phi}_s  :=	\left (\int \abs{\phi}^s \right)^{\frac 1s}
						\qquad (s \ge 1).
\]
Define also the following family of pseudometrics in~$\A$
\[
	\rho_s(\alg A, \alg B) :=\norm{ L(\alg A) - L(\alg B) }_s.
\]
		\begin{rmk}
	Consider any $1\le s < \infty$ and assume that the topology
	in $\Img L$ comes from the norm~$\norm{\cdot}_s$.
	Then the mapping $L\colon \A \to \Img L$ is continuous.
	In particular, for any additive partition entropy~$I$ if
	the mapping $J\colon \Img L \to \R$, defined by
	$I = J \circ L$, is continuous then the additive partition entropy~$I$ is
	continuous as well.\end{rmk}

Indeed, fix $\alg A = \indalg{A_1, \dotsc, A_n}$ and
$0 < \varepsilon < 1$.
Put $M := \max \{\mu(A_i)\}$, $m := \min \{\mu(A_i) \neq 0\}$.
There is $0 < \delta < \min \{ \varepsilon, m/2 \}$ such that
whenever $\mu(A \symdiff B) < \delta$ then also
$\abs{ \log \mu(A) - \log \mu(B) } < \varepsilon$.
If now $d(\alg A, \alg B) < \delta < 1$ then the algebras $\alg A$ and $\alg B$
have the same number of atoms of nonzero measure. We can
suppose that $\alg B = \indalg{B_1, \dotsc, B_n}$ and $\mu(A_i), \mu(B_i) > 0$,
and also that $\mu( A_i \symdiff B_i )  <  \delta$. Then
\begin{align*}
	\norm{L(\alg A) - L(\alg B)}_s^s
	&	= \sum_{i,j} \mu(A_i \cap B_j) \abs{ \log \tfrac 1{\mu(A_i)}
			- \log \tfrac 1{\mu(B_j)} }^s	\\
	&	<  \sum_i \mu(A_i \cap B_i) \varepsilon^s
		+ \sum_{i \neq j} \delta \abs{ \log \tfrac 1{\mu(A_i)}
			- \log \tfrac 1{\mu(B_j)} }^s	\\
	&	< nM\varepsilon + (n^2-n)\varepsilon
			\left( \log \tfrac 1m + \log \tfrac 2m \right)^s \\
	&	= \varepsilon \tim \mathrm{const}( \alg A ).
\end{align*}

		\begin{crr}		\label{crr:cty}
	Every additive partition entropy~$I$ continuous with respect to any
	pseudometric~$\rho_s$ ($1 \le s < \infty$) is continuous.\end{crr}

Observe, in addition, that by the power mean inequality, when $s>r$ the topology
generated by the pseudometric $\rho_s$ is
richer from the one generated by $\rho_r$. In particular additive partition entropies
continuous with respect to $\rho_r$ are continuous with respect to $\rho_s$.

%

\end{ver:old}

	\section {Final remarks}

At first, let us mention some other possible approaches one might take to derive our results. The first prototype of the proof of \mref {Lemma} {lem:keycts}
was based around defining $\widehat \m(\charf V - \charf W) := \Delta(V,W)$,
extending $\widehat \m$ to the linear space generated by
functions of shape~$\charf V$, observing that $\widehat \m$ is a continous linear mapping,
and then extending $\widehat \m$ to the space of all continuous
functions and finally defining $\m(V) :=\widehat \m (\charf V)$.

One could try to prove \mref {Theorems} {thm:cts} 
by elaborating on \mref {Remark} {rmk:H} and a version of Proposition \ref
{prn:same}. Although unexplored, it seems that such an approach would suffer
from some problems of its own.

\begin{ver:old}
This section contains some simple and known properties of measure. Their
inclusion will have little impact on the length of the paper.
\end{ver:old}
\begin{ver:new}
We would like now to convert our theorem to the typical case of a countably-additive probability. To do so, we need the following simple and known properties of measure. The proofs are omitted.
\end{ver:new}

		\begin{rmk}
	If $\mu$ is a countably-additive measure on $\sigma${\dash}algebra $\alg F$
	then every finitely-additive set function~$\m$ absolutely continuous
	with respect to~$\mu$
	is countably-additive.
	\end{rmk}

\begin{ver:old}
		\begin{proof}
	In fact, consider any sequence of disjoint sets
	$A_i \in \alg F$. Using the convergence of~$\mu(\bigdissum_{i>n} A_i) =
	\sum_{i>n} \mu(A_i) \tendsto 0$ we get
	\[
		\abs[\bigg]{ \m\Big(\bigdissum_{i \ge 1} A_i\Big) - \sum_{i \le n} \m(A_i) }
		= \abs[\bigg]{ \m\Big(\bigdissum_{i>n} A_i\Big) } \tendsto 0.
	\]	\end{proof}
\end{ver:old}

		\begin{rmk}		\label{rmk:darbatom}
	Every nonatomic probabilty space satisfies the
	\hyperref[rmk:assumption] {Darboux condition}. \end{rmk}

\begin{ver:old}
		\begin{proof}
	The assumption means that for any $A$, $P(A) > 0$ we can find
	$B$, $C$ such that $P(B),P(C) > 0$ and $A = B \dissum C$.

	Set $\eta := \inf_{A : P(A) \ge 1/2} P(A)$.
	Observe that for any $\varepsilon > 0$ and a measurable set $B$
	there is a set~$A\subset B$ such that $0< P(A) \le \varepsilon $. Hence
	if the infimum in the definition of~$\eta$ is attained
	it is~$1/2$.

	Assume that $\eta \ge 3/4$. Then $\eta$  is not attained
	and there exists a sequence of sets $A_1, A_2, \dotsc$ with measures greater than $\eta$
	and tending to~$\eta$.
	Now, set $B_n := A_1 \cap \dotsb \cap A_n$. A simple induction shows that
	$P(B_n) > 1/2$ and so $P(B_n) > \eta$. What is more,
	$B_n \supset B_{n+1}$ and $P(B_n) < P(A_n)$ i.e.
	$P(B_n) \tendsto \eta$. Then $P(\bigcap B_n) = \eta$
	--- a contradiction. Thus $1/2 \le \eta < 3/4$. Since we have the same fact
	for every subset $A \subset \Omega$,
	with $P(A) > 0$ as a conditional probability space,
	then we can divide every such set~$A$  into two parts such that
	the measure of every of them is over $4/3$ times smaller than the
	measure of~$A$. Because of $(3/4)^n \tendsto 0$ and using increasingly
	fine divisions we can find for any~$0 < \theta < P(A) $ an ascending
	sequence of sets $C_n$ with their measures tending to~$\theta$.
	Their union satisfies $P(\bigcup C_n) = \theta$.	\end{proof}
\end{ver:old}

In view of these remarks, we can recast \mref{Theorem} {thm:cts} in the setting of
a nonatomic probability space $(\Omega, \alg \Sigma, P)$.

		\begin{thm}		\label{thm:probmain}
	Let $I$ be a additive partition entropy, continuous \new{on $\A_2$ (c.f. \mref{Definition} {dfn:Icts},
	and \mref{Definition} {dfn:IctsA2})}. There exists an additive entropy, and a countably-additive
	set funcion~$\m$ absolutely continuous with respect to probability~$P$ such that
	\[
		I=H_P + L_\m,
	\]
	There is only one such pair with $\m(\Omega)=0$.	$L_\m$ is continuous. If $I$ is continuous then also $H_P$ is continuous, i.e. $H$ is continuous in the open domain. (c.f. \mref{Remark}{rmk:seqentrcts})\end{thm}

\new{We would like to mention that this result will be transferred to a quantum context in a follow-up paper,~\cite{PS}.}

\begin{ver:old}
To complete the picture we still need a description of the additive partition entropies that
depend solely on the measures of atoms. Clearly any such description
does not depend on the choice of a probability space.
A.~Paszkiewicz has proved the following two Theorems. These results
have not been published so far.

		\begin{thm}		\label{thm:paszk}
	If $\widetilde I$ is an additive partition entropy which is uniformly continuous
	with respect to~$\rho_2$ and depends solely on the measures of atoms
	then there exist numbers $\alpha, \beta$ such that
	\[
		\widetilde I = \alpha L_\mu + \beta V,
	\]
	where $L_\mu$ is the Shannon entropy (see \mref{Example}{exm:H}) and
	$V$ is the additive partition entropy of \mref{Example}{exm:V}. \end{thm}

This leads on to the following general characterisation of all additive partition entropies wich are
uniformly continuous with respect to~$\rho_2$.

		\begin{thm}
	If $I$ is an additive partition entropy uniformly continuous with respect to~$\rho_2$,
	then there is a countably-additive set function $\m$
	absolutely continuous with respect to~$\mu$ and a number $\beta$ such that
	\begin{equation}		\label{eqn:mainpaszk}
		I(\alg A) = L_\m(\alg A) + \beta V(\alg A) =
		\int \! H(\alg A) \, d\m + \beta \! \int\! \Bigl[
			H(\alg A) -  \int \! H(\alg A) \, d\mu
		\Bigr]^2 d\mu.
	\end{equation}	\end{thm}

		\begin{proof}
	It follows directly from \mref{Theorem}{thm:probmain},
	\mref{Theorem}{thm:paszk} and \mref {Corollary} {crr:cty}. \end{proof}
\end{ver:old}

%


%

\begin{ver:old}

	\section*{Acknowledgements}

I would like to express my deepest gratitude for my mentor professor Adam Paszkiewicz
who proposed the problem of the paper and also devoted a lot of time discussing
it. I owe him a  huge part of work. Perhaps his most important contribution
is providing at a crucial moment \mref{Example}{exm:H}, which does not depend on measures
of atoms and which lead him to formulate a version of the main result of this paper
together with some rough idea about its proof.
Professor Paszkiewicz has also modified the result that came to be
\mref{Proposition}{prn:same}, proved \mref {property} {prop:cont}
in \mref {Proposition} {prn:delta}, supplied me with~\mref{Example}{exm:V}, and
formulated \mref {Remark} {rmk:darbatom}. To top it all,
he went through and verified the paper several times,
and had a pronounced influence on its editing.

I give my heartfelt thanks to my wife Maria and daughter Amelia;
my wife for the time I could have sacrificed working on the paper; my daughter
for the tranquil, well-slept nights. The paper is devoted to Amelka,
who was born shortly after I started thinking on the problem.
\end{ver:old}


\end{document}